\documentclass[a4paper,UKenglish,cleveref, autoref, thm-restate]{lipics-v2021}

\hideLIPIcs  

\usepackage{xspace}
\usepackage{stmaryrd}
\usepackage[ruled,vlined]{algorithm2e}

\title{A faster polynomial-space algorithm for Hamiltonian cycle parameterized by treedepth} 

\author{Stefan Kratsch}{Humboldt-Universit\"at zu Berlin, Berlin, Germany}{stefan.kratsch@hu-berlin.de}{https://orcid.org/0000-0002-1825-0097}{}

\authorrunning{S. Kratsch} 

\Copyright{Stefan Kratsch} 



\category{} 

\relatedversion{} 

\nolinenumbers 

\EventEditors{John Q. Open and Joan R. Access}
\EventNoEds{2}
\EventLongTitle{42nd Conference on Very Important Topics (CVIT 2016)}
\EventShortTitle{CVIT 2016}
\EventAcronym{CVIT}
\EventYear{2016}
\EventDate{December 24--27, 2016}
\EventLocation{Little Whinging, United Kingdom}
\EventLogo{}
\SeriesVolume{42}
\ArticleNo{23}

\newcommand{\integers}{\ensuremath{\mathbb{Z}}\xspace}
\newcommand{\naturals}{\ensuremath{\mathbb{N}}\xspace}

\newcommand{\tw}{\mathrm{tw}}
\newcommand{\pw}{\mathrm{pw}}
\newcommand{\td}{\mathrm{td}}

\newcommand{\C}{\ensuremath{\mathcal{C}}\xspace}
\newcommand{\Oh}{\mathcal{O}}
\newcommand{\T}{\ensuremath{\mathcal{T}}\xspace}

\newcommand{\tree}{\mathtt{tree}}
\newcommand{\tail}{\mathtt{tail}}
\newcommand{\children}{\mathtt{children}}
\newcommand{\broom}{\mathtt{broom}}

\newcommand{\monomial}{\mathtt{mon}}
\newcommand{\tuple}{\mathtt{tpl}}
\newcommand{\disjoint}{\mathtt{dsj}}

\newcommand{\ComputeP}{\texttt{ComputeP}\xspace}

\newcommand{\yes}{\texttt{yes}\xspace}
\newcommand{\no}{\texttt{no}\xspace}

\newcommand{\iverson}[1]{\ensuremath{\llbracket #1 \rrbracket}}

\newcommand{\wdpm}{\ensuremath{\mathcal{M}_{w,\ell}}}

\newcommand{\weights}{\ensuremath{\mathbf{w}}\xspace}

\newcommand{\problem}[1]{\textsc{#1}\xspace}
\newcommand{\qcoloring}{\problem{$q$-Coloring}}

\newcommand{\maxcut}{\problem{Max Cut}}

\newcommand{\longcycle}{\problem{Long Cycle}}

\newcommand{\deletiontorcolorable}{\problem{Deletion To $r$-Colorable}}
\newcommand{\dominatingset}{\problem{Dominating Set}}
\newcommand{\independentset}{\problem{Independent Set}}
\newcommand{\hamiltoniancycle}{\problem{Hamiltonian Cycle}}

\newcommand{\longpath}{\problem{Long Path}}
\newcommand{\longestpath}{\problem{Longest Path}}
\newcommand{\hamiltonianpath}{\problem{Hamiltonian Path}}
\newcommand{\oddcycletransversal}{\problem{Odd Cycle Transversal}}
\newcommand{\partialcyclecover}{\problem{Partial Cycle Cover}}
\newcommand{\mincyclecover}{\problem{Min Cycle Cover}}
\newcommand{\qcnfsat}{\problem{$q$-CNF Sat}}
\newcommand{\tsp}{\problem{Traveling Salesperson}}
\newcommand{\vertexcover}{\problem{Vertex Cover}}

\newcommand{\class}[1]{\ensuremath{\mathsf{#1}}\xspace}
\newcommand{\classP}{\class{P}}
\newcommand{\NP}{\class{NP}}

\begin{document}

\maketitle

\begin{abstract}
 A large number of \NP-hard graph problems can be solved in $f(w)n^{\Oh(1)}$ time and space when the input graph is provided together with a tree decomposition of width $w$, in many cases with a modest exponential dependence $f(w)$ on $w$. Moreover, assuming the Strong Exponential-Time Hypothesis (SETH) we have essentially matching lower bounds for many such problems. They main drawback of these results is that the corresponding dynamic programming algorithms use exponential space, which makes them infeasible for larger $w$, and there is some evidence that this cannot be avoided.
 
 This motivates using somewhat more restrictive structure/decompositions of the graph to also get good (exponential) dependence on the corresponding parameter but use only polynomial space. A number of papers have contributed to this quest by studying problems relative to treedepth, and have obtained fast polynomial space algorithms, often matching the dependence on treewidth in the time bound. E.g., a number of connectivity problems could be solved by adapting the cut-and-count technique of Cygan et al.\ (FOCS 2011, TALG 2022) to treedepth, but this excluded well-known path and cycle problems such as \hamiltoniancycle (Hegerfeld and Kratsch, STACS 2020).
 
 Recently, Nederlof et al.\ (SIDMA 2023) showed how to solve \hamiltoniancycle, and several related problems, in $5^\tau n^{\Oh(1)}$ randomized time and polynomial space when provided with an elimination forest of depth $\tau$. We present a faster (also randomized) algorithm, running in $4^\tau n^{\Oh(1)}$ time and polynomial space, for the same set of problems. We use ordered pairs of what we call \emph{consistent} matchings, rather than perfect matchings in an auxiliary graph, to get the improved time bound.
\end{abstract}

\section{Introduction}
\label{section:introduction}

It is widely believed that \NP-hard problems do not admit polynomial-time algorithms for solving them exactly (i.e., that $\classP\neq\NP$). A core question of parameterized complexity is, therefore, to what extent input structure can be leveraged to nevertheless get reasonably fast algorithms, i.e., \emph{how structure affects complexity}. It is well known, for example, that a large number of \NP-hard graph problems can be solved even in linear time on graphs of bounded treewidth, i.e., in time $f(\tw)(n+m)$ via Courcelle's theorem~\cite{DBLP:journals/iandc/Courcelle90} and its extensions. While the function $f(\tw)$ is enormous in general, much better bounds were obtained for a variety of problems.\footnote{Often at the cost of a modest increase in the polynomial dependence on the graph size, and with the best bounds assuming that a tree decomposition of the graph is already known.} More recently, initiated by Lokshtanov et al.~\cite{DBLP:journals/talg/LokshtanovMS18}, it was shown that many of these improved bounds were in fact optimal under the Strong Exponential-Time Hypothesis (SETH),\footnote{SETH is the hypothesis that for all $\varepsilon>0$ there is $q\in\naturals$ such that \qcnfsat cannot be solved in $\Oh((2-\varepsilon)^n)$ time for formulas with $n$ variables.} e.g., \dominatingset can be solved in $3^{\tw}n^{\Oh(1)}$ but not in $(3-\varepsilon)^{\tw}n^{\Oh(1)}$ time for any $\varepsilon>0$ (assuming SETH). In some cases, arriving at such tight bounds first required novel algorithmic tools, such as the \emph{cut-and-count} technique of Cygan et al.~\cite{DBLP:journals/talg/CyganNPPRW22}, which enabled single-exponential dependence on the treewidth for many connectivity-related problems.

With these much better (and possibly optimal) bounds relative to the treewidth there still remains a caveat: All these algorithms rely on some form or other of dynamic programming (DP) and they essentially use as much space as time, i.e., exponential space relative to treewidth (and ditto for other width parameters). While the (conditionally) optimal bounds are very satisfying, this is quite an obstacle for applying the algorithms on a wider scale, as a lack of memory stops a computation whereas time is (somewhat) more readily available.\footnote{In the words of Gerhard Woeginger~\cite{DBLP:conf/iwpec/Woeginger04}: ``Note that algorithms with exponential space complexities are absolutely useless for real life applications.''} Unfortunately, this does not come from a lack of algorithmic design capability, but it is widely believed that these time bounds are not obtainable using significantly less space. Moreover, there are both unconditional lower bounds and completeness for specific models~\cite{DBLP:conf/esa/DruckerNS16,DBLP:journals/algorithms/ChenRRV18} as well as general, conditional lower bounds~\cite{DBLP:journals/toct/PilipczukW18}. Thus, to get polynomial space, we will have to reduce the generality of the input structure that we can deal with, i.e., we need to use more restrictive graph parameters corresponding to more restrictive structure and decompositions.

A key parameter in this regard is the \emph{treedepth} of the graph (along with several generalizations): The treedepth of a graph $G=(V,E)$ is defined as the minimum depth of a rooted forest $\T=(V,F)$ such that for each edge $\{u,v\}\in E$ one of $u$ and $v$ is an ancestor of the other in $\T$. Such a forest $\T$ is then also called an \emph{elimination forest} of $G$, which comes from a different but equivalent way of defining treedepth. It is well known (and easy to see) that treedepth is at least as big as pathwidth, which is at least as big as treewidth; at the same time, the treedepth of a graph with $n$ vertices is at most $\log n$ times its treewidth. Where path and tree decompositions lend themselves to dynamic programming (using exponential space), elimination forests also often allow branching algorithms with the same parameter dependence but using only polynomial space. E.g., \dominatingset can be solved in $3^{\td}n^{\Oh(1)}$ time and polynomial space for graphs of treedepth $\td$ \cite{DBLP:journals/toct/PilipczukW18}.

Generally, branching seems the method of choice to get fast polynomial-space algorithms, yet even simple dynamic programming approaches do not easily carry over to branching: Intuitively, in partial solutions maintained by DP algorithms we usually have changing roles/states of vertices, e.g., a vertex turning from undominated to dominated, which at first glance is hard to do in a branching algorithm that essentially gives each vertex a fixed state (though will make that ``choice'' many times in its run). Only for what one may call ``static'' partial solutions, e.g., for \qcoloring, is there an immediate way to transfer from DP to branching. In this way, given an elimination forest of depth $\tau$, it is straightforward to solve \independentset and \qcoloring in $2^\tau n^{\Oh(1)}$ resp.\ $q^\tau n^{\Oh(1)}$ time and polynomial space. 

F\"urer and Yu~\cite{DBLP:journals/mst/FurerY17} adapted the algebraization technique of Lokshtanov and Nederlof~\cite{DBLP:conf/stoc/LokshtanovN10} to a dynamic setting to obtain an algorithm that counts the number of perfect matchings in $2^\tau n^{\Oh(1)}$ time and polynomial space. Chen et al.~\cite{DBLP:journals/algorithms/ChenRRV18} (next to lower bounds above) solve \dominatingset in either $\tau^{\Oh(\tau^2)}n$ time and $\tau^{\Oh(1)}\log n$ space or in $3^\tau\log\tau n$ time and $\Oh(2^\tau\tau\log\tau+\tau\log n)$ space. Building on F\"urer and Yu~\cite{DBLP:journals/mst/FurerY17}, Pilipczuk and Wrochna~\cite{DBLP:journals/toct/PilipczukW18} design a $3^\tau n^{\Oh(1)}$ time and polynomial space algorithm for \dominatingset, and then improve its space usage to as little as $\Oh(\tau \log n)$ using Fourier transform and Chinese remaindering (next to foundational work above). Independently, Belbasi and F\"urer~\cite{DBLP:journals/corr/abs-1711-10088} got the same time bound but use a little more space. A different work of Belbasi and F\"urer~\cite{DBLP:conf/csr/BelbasiF19} solves \hamiltoniancycle in $(4w)^\tau n^{\Oh(1)}$ time and $\Oh(w\tau n\log n)$ space where $w$ is the width and $\tau$ the depth of a given tree decomposition.\footnote{F\"urer and Yu~\cite{DBLP:journals/mst/FurerY17} point out that treedepth can also be characterized as the maximum number of forget nodes in any nice tree decomposition of the graph, which they call the depth of the decomposition.} The same time but with an extra factor of sum of edge weights in space is also obtained for \tsp~\cite{DBLP:conf/csr/BelbasiF19}. 

Hegerfeld and Kratsch~\cite{DBLP:conf/stacs/HegerfeldK20} adapted the cut-and-count technique of Cygan et al.~\cite{DBLP:journals/talg/CyganNPPRW22} to obtain, for a number of connectivity problems, the same (randomized) time bound as relative to treewidth but using only polynomial space. The main omission in their work, and seemingly less amenable to cut-and-count relative to treedepth, are the path and cycle problems covered by Cygan et al.\ such as \hamiltoniancycle and \mincyclecover. This was addressed by Nederlof et al.~\cite{DBLP:journals/siamdm/NederlofPSW23} who used a different approach via perfect matchings in an auxiliary to solve \partialcyclecover in $5^\tau n^{\Oh(1)}$ randomized time and polynomial space; in this problem, given graph $G$ and numbers $k,\ell$ we need to decide whether exactly $\ell$ vertices of $G$ can be covered with at most $k$ vertex-disjoint cycles. From this, they directly get the same bounds for a bunch of related problems such as \hamiltoniancycle, \longpath, and \mincyclecover. They point out the gap to the dependence of $(2+\sqrt{2})^p n^{\Oh(1)}$ time and space relative to the width $p$ of a given path decomposition~\cite{DBLP:journals/jacm/CyganKN18} and ask whether that time can be matched relative to treedepth but using polynomial space only. That being said, any improvement to the base $5$ in the dependence on $\tau$ would be of interest~\cite{DBLP:journals/siamdm/NederlofPSW23}.

\subparagraph{Our work.}
We design a faster polynomial space Monte-Carlo algorithm for the \partialcyclecover problem (and thereby also for \hamiltoniancycle). Our algorithm runs in $4^{\tau}n^{\Oh(1)}$ time and polynomial space when provided with an elimination forest of depth $\tau$ for the graph. Like the algorithm of Nederlof et al.~\cite{DBLP:journals/siamdm/NederlofPSW23} it does not make false positives and the chance of a false negative is at most $\frac12$. As a corollary, observed by Nederlof et al.~\cite{DBLP:journals/siamdm/NederlofPSW23}, we get the same bounds also for a number of classical path and cycle problems such as \hamiltoniancycle and \longestpath.

\begin{theorem}
 \label{theorem:main:partialcyclecover}
 There is a $4^{\tau}n^{\Oh(1)}$ time and polynomial space Monte-Carlo algorithm that, given a graph $G$, numbers $k,\ell\in\naturals$, and an elimination forest of depth $\tau$ for $G$, solves \partialcyclecover in $4^{\tau}n^{\Oh(1)}$ time and polynomial space. Errors are limited to false negatives and occur with probability at most $\frac12$.
\end{theorem}

\begin{corollary}
 \label{corollary:main:applications}
 Given a graph $G$, and $k,\ell$ as needed, and an elimination forest of depth $\tau$ for $G$, one can solve \hamiltoniancycle, \hamiltonianpath, \longcycle, \longpath, and \mincyclecover in $4^{\tau}n^{\Oh(1)}$ time and polynomial space with false negatives only and error chance at most $\frac12$.
\end{corollary}

Our algorithm crucially relies on properties of what we call consistent pairs of matchings: Say that two matchings $M_1$ and $M_2$ are \emph{consistent} if (1) $M_1\cap M_2=\emptyset$ and (2) $V(M_1)=V(M_2)$. Using well known relations between matchings and (partial) cycle covers one immediately observes that the union of two matchings $M_1$ and $M_2$ in $G$ is a partial cycle cover in $G$ if and only if $M_1$ and $M_2$ are consistent. Note, however, that the resulting partial cycle covers can only have even-length cycles. Fortunately, we can easily reduce \partialcyclecover to the special case where $G$ is bipartite, while only adding $1$ to the depth $\tau$, so that all partial cycle covers consist of even-length cycles only.

Our main technical contribution is a $4^{\tau}n^{\Oh(1)}$ time and polynomial space deterministic algorithm that, given a graph $G=(V,E)$, $w,\ell\in\naturals$, edge weights $\weights\colon E\to\{1,\ldots,2\ell\}$, and a an elimination forest of depth $\tau$ for $G$, computes the the number of \emph{ordered pairs} of consistent matchings of cardinality $\frac\ell2$ each and total edge weight $w$, which we denote by $|\wdpm|$. This itself is based on an inclusion-exclusion formula for $|\wdpm|$ with $2n$ requirements.

To come back to partial cycles covers, it is easy to see (and detailed later) that a partial cycle cover with exactly $p$ cycles visiting exactly $\ell$ vertices corresponds to $2^p$ ordered pairs of consistent matchings of cardinality $\frac\ell2$, which is similar to the symmetry used by Nederlof et al.~\cite{DBLP:journals/siamdm/NederlofPSW23}. By a standard application of the isolation lemma, we can reduce (with at least $\frac12$ chance of success) to the case where there either is no partial cycle cover with at most $k$ cycles on exactly $\ell$ vertices, or there is a unique such partial cycle cover of minimum total edge weight. Thus, counting the ordered pairs of consistent matchings modulo $2^{k+1}$ will let us detect which of the two cases is true, as all partial cycle covers with $p\geq k+1$ cycles contribute a multiple of $2^{k+1}$, which is congruent to $0$.

\subparagraph{Related work.}
Recently, Bergougnoux et al.~\cite{DBLP:conf/soda/BergougnouxCS26} presented a logic-based meta-theorem for problems parameterized by treedepth: They show $2^{\Oh(\tau)} n^{\Oh(1)}$ time and polynomial space for problems expressible in what they call neighborhood operator logic with acyclicity, connectivity, and clique constraints. Dropping acyclicity and connectivity lets them improve the space to $\Oh(\tau \log n)$. This interesting result further establishes treedepth for obtaining fast parameterized algorithms with polynomial space, while still leaving room for hunting down better or even conditionally optimal time bounds (modulo SETH) for specific problems.

Speaking of conditional optimality (hence lower bounds), Jaffke and Jansen~\cite{DBLP:journals/dam/JaffkeJ23} showed that (modulo SETH) there is no $(q-\varepsilon)^{|X|}n^{\Oh(1)}$ time algorithm for \qcoloring when provided with $X$ such that $G-X$ is a disjoint union of paths; this implies the same lower bound relative to treedepth using that $X$ gives rise to an elimination forest of depth $|X|+\log n$. Koutis et al.~\cite{DBLP:conf/iwpec/Koutis0Z23} showed that under either of SETH or the Set Cover Conjecture (see~\cite{DBLP:journals/talg/CyganDLMNOPSW16}) \dominatingset parameterized by the size $vc$ of a given vertex cover cannot be solved in $\Oh((2-\varepsilon)^{vc}n^{\Oh(1)}$ time (and gave a matching algorithm relative to $vc$); the lower bound carries over to \dominatingset relative to treedepth but the best known algorithm takes $3^{\tau}n^{\Oh(1)}$ time. Hegerfeld and Kratsch~\cite{DBLP:conf/iwpec/HegerfeldK22} gave several lower bounds that carry over relative to treedepth, e.g., that (modulo SETH) there is no $(r+1-\varepsilon)^\tau n^{\Oh(1)}$ time algorithm for \deletiontorcolorable, which generalizes \vertexcover and \oddcycletransversal. This was pushed further by Esmer et al.~\cite{DBLP:conf/icalp/EsmerFMR24} who cover also some packing problems and \dominatingset,\footnote{Notably, for \dominatingset this does match base $3$ as for treewidth, but only rules out base $2-\varepsilon$.} and with a bit less structure required, but it remains that reproducing the bases relative to treewidth works best for ``coloring-like'' problems. E.g., we know nothing about lower bounds for \hamiltoniancycle relative to treedepth.

We would be remiss not to mention the state of the art for actually finding optimal elimination forests, i.e., with depth $\tau=\td(G)$ (or testing whether $\td(G)\leq\tau$, given $\tau$): Reidl et al.~\cite{DBLP:conf/icalp/ReidlRVS14} were able to do this in $2^{\Oh(\tau^2)}\cdot n$ time and exponential space. More recently, Nadara et al.~\cite{DBLP:conf/esa/NadaraPS22} showed $2^{\Oh(\tau^2)}\cdot n^{\Oh(1)}$ time and polynomial space, but also $2^{\Oh(\tau^2)}\cdot n$ randomized time and $\tau^{\Oh(1)}\cdot n$ space. As pointed out by Nederlof et al.~\cite{DBLP:journals/siamdm/NederlofPSW23}, we are ``missing'' a constant-factor approximation in $2^{\Oh(\tau)}n^{\Oh(1)}$ time, ideally with polynomial space.

Bergougnoux et al.~\cite{DBLP:journals/toct/BergougnouxCGKMOPL25} showed that the existence of fast polynomial-space algorithms can be pushed beyond treedepth by presenting single-exponential time and polynomial-space algorithms for \independentset, \dominatingset, and \maxcut relative to shrubdepth, which is a dense generalization of treedepth and similar also to cliquewidth.\footnote{We omit the time bounds here since they only make sense after introducing the necessary \emph{tree models}.}

\subparagraph{Organization.}
Section~\ref{section:preliminaries} provides the necessary notation and recalls the isolation lemma. In Section~\ref{section:ieformula} we derive an inclusion-exclusion formula for the number $|\wdpm|$ of ordered pairs of consistent matchings of given cardinality and weight.
Section~\ref{section:computingwdpm} is the longest and most technical part, and shows how to recursively compute $|\wdpm|$ given an elimination forest for the graph.
Section~\ref{section:applications} shows the algorithm for \partialcyclecover and the implied further applications.
We conclude in Section~\ref{section:conclusion} with a few questions for further research.

\section{Preliminaries}
\label{section:preliminaries}

We agree that $\naturals=\{0,1,\ldots\}$. We use $[n]:=\{1,\ldots,n\}$ and $[a,b]:=\{a,a+1,\ldots,b\}$. We use $\iverson{\ldots}$ for Iverson's bracket notation, where, for a proposition $p$, we have $\iverson{p}=1$ if $p$ is true, and $\iverson{p}=0$ otherwise. Note, e.g., that $\iverson{p\wedge q}=\iverson{p}\iverson{q}$.

\subparagraph{Graphs.}
All graphs in this work are undirected and simple. We use standard graph notation (cf.~\cite{DBLP:books/daglib/0030488}). For a graph $G=(V,E)$ and $U\subseteq V$ we denote by $\delta[U]$ the set of edges having \emph{at least} one endpoint in $U$, formally $\delta[U]=\{e\in E\mid e\cap U\neq\emptyset\}$. A \emph{matching} in $G=(V,E)$ is a set $M\subseteq E$ of pairwise disjoint edges. We recall (from the introduction) that two matchings $M_1$ and $M_2$ are \emph{consistent} if (1) $M_1\cap M_2=\emptyset$ and (2) $V(M_1)=V(M_2)$. A \emph{partial cycle cover} in $G=(V,E)$ is a set $C\subseteq E$ of edges such that each vertex of $(V,C)$ has degree zero or two, i.e., it corresponds to a family of vertex-disjoint (simple) cycles in $G$.

\subparagraph{Treedepth.}
Let $G=(V,E)$ be a graph. An \emph{elimination forest} $\T$ for $G$ is a rooted forest on the same vertex set $V$ as $G$ (but allowed any edges on $G$) such that for each edge $\{u,v\}\in E$ either $u$ is an ancestor of $v$ in $T$ or vice versa. The \emph{depth} $\tau$ of an elimination forest is the largest number of vertices (not edges) from any root of the forest to a leaf (of its tree). The \emph{treedepth} of $G$ is the smallest depth $\tau$ over all elimination forests for $G$. Though $G$ and $\T$ have the same vertex set, we shall speak of vertex $v$ when referring to $v$ as a vertex of $G$, and of node $v$ when referring to its role in $\T$. E.g., $v$ can have child nodes (as a node of $\T$).

We use treedepth-related notation from~\cite{DBLP:journals/toct/PilipczukW18}, slightly extended by \cite{DBLP:conf/stacs/HegerfeldK20}: For a rooted forest $\T=(V,F)$ and node $v\in V$, we denote by $\tree[v]$ the set of nodes in the subtree of $v$ in $\T$, including $v$. Furthermore, we let $\tree(v):=\tree[v]\setminus\{v\}$. By $\tail[v]$ we denote the set of nodes that are ancestors of $v$, including $v$. Furthermore, we let $\tail(v):=\tail[v]\setminus\{v\}$. Finally, we use $\broom[v]:=\{v\}\cup\tail(v)\cup\tree(v)$, and we will not define or use $\broom(v)$. 

\subparagraph{Isolation lemma.}
We recall the well known isolation lemma due to Mulmuley et al.~\cite{DBLP:journals/combinatorica/MulmuleyVV87}:

\begin{lemma}[Isolation Lemma, \cite{DBLP:journals/combinatorica/MulmuleyVV87}]
 Let $\F\subseteq 2^U$ be a nonempty set family over a universe $U$ and let $N\in\naturals$. Let $\weights\colon U\to[N]$ be chosen uniformly and independently at random. Then with probability at least $1-\frac{|U|}{N}$ there is a unique $F\in\F$ of minimum size $\weights(F):=\sum_{u\in F}\weights(u)$.
\end{lemma}

\section{Inclusion-exclusion for ordered pairs of consistent matchings}
\label{section:ieformula}

Throughout this section let $G=(V,E)$ be a graph with vertex set $V=\{1,\ldots,n\}$ let $\weights\colon E\to\{1,\ldots,N\}$ be a weight function, and let $w,\ell\in\naturals$, with $\ell$ even and $\ell\leq n$. We will derive an inclusion-exclusion formula for the number of ordered pairs of consistent matchings of cardinality $\frac\ell2$ and total edge weight equal to $w$. Formally, this is equal to $|\wdpm|$ where
\begin{align*}
 \wdpm:= \{ (M_1,M_2) \mid{} & \mbox{$M_1$ and $M_2$ are consistent matchings in $G$}, |M_1|=|M_2|=\ell/2, \\
 & \weights(M_1\cup M_2)=w \}.
\end{align*}

We will define a ground set $U_{w,\ell}$ and requirements $A_{w,\ell,i}\subseteq U_{w,\ell}$, for $i\in[2n]$ such that
\begin{align}
 \label{math:wdpmviareq}
 \left| \bigcap_{i\in [2n]} A_{w,\ell,i} \right| = |\wdpm|.
\end{align}
The set $U_{w,\ell}$ is defined to contain all triples $(E_1,E_2,L)$ where $E_1$ and $E_2$ are disjoint edge subsets of $G$ of cardinality $\frac\ell2$ each and total weight equal to $w$, and $L$ is a vertex subset of $G$ of cardinality $n-\ell$. Formally,
\begin{align*}
 U_{w,\ell}:= \left\{ (E_1,E_2,L) \mid E_1,E_2\in \binom{E}{\ell/2}, E_1\cap E_2=\emptyset, \weights(E_1\cup E_2)=w, L\in\binom{V}{n-\ell} \right\}.
\end{align*}

We define $2n$ requirements $A_{w,\ell,i}\subseteq U_{w,\ell}$. The first $n$ enforce coverage of each $i\in V$ by $E_1$ or $L$, the second $n$ enforce coverage by $E_2$ or $L$. Formally,
\begin{align*}
 A_{w,\ell,i}:= \begin{cases}
                 \left\{ (E_1,E_2,L)\in U_{w,\ell} \mid i\in V(E_1)\cup L \right\} & \mbox{if $i\in[n]$}\\
                 \left\{ (E_1,E_2,L)\in U_{w,\ell} \mid i-n\in V(E_2)\cup L \right\} & \mbox{if $i\in[n+1,2n]$.}
                \end{cases}
\end{align*}

We show that this setup fulfills $(\ref{math:wdpmviareq})$.

\begin{lemma}
 $\left| \bigcap_{i\in [2n]} A_{w,\ell,i} \right| = |\wdpm|$.
\end{lemma}

\begin{proof}
 Let $(E_1,E_2,L)\in \bigcap_{i\in [2n]} A_{w,\ell,i}$. By definition of $A_{w,\ell,i}$ we have $V\subseteq V(E_1)\cup L$ and $V\subseteq V(E_2)\cup L$. Since $|E_1|=|E_2|=\frac\ell2$, we have $|V(E_1)|\leq \ell$ and $|V(E_2)|\leq\ell$. Since furthermore $|L|=n-\ell$ and $|V|=n$, it follows that $V(E_1)=V(E_2)=V\setminus L$ and $|V(E_1)|=|V(E_2)|=\ell$. Hence, $E_1$ and $E_2$ are matchings on the same vertex set. Thus, all triplets in $\bigcap_{i\in [2n]} A_{w,\ell,i}$ consist of two such matchings together with the set~$L$ of $n-\ell$ uncovered vertices. Since $w=\weights(E_1,E_2)$, we have $(E_1,E_2)\in\wdpm$. Thus, for all $(E_1,E_2,L)\in\bigcap_{i\in [2n]} A_{w,\ell,i}$ we have $(E_1,E_2)\in\wdpm$. For any two different $(E_1,E_2,L),(E'_1,E'_2,L)\in\bigcap_{i\in [2n]} A_{w,\ell,i}$ we must have $E_1\neq E'_1$ or $E_2\neq E'_2$, or else by the above it follows that $L=L'$. But then such different triplets correspond to different pairs $(E_1,E_2)$ and $(E'_1,E'_2)$ in $\wdpm$, which implies that $\left| \bigcap_{i\in [2n]} A_{w,\ell,i} \right| \leq |\wdpm|$.
 
 Now, let $(M_1,M_2)\in\wdpm$. Let $L:=V\setminus V(M_1)$. Clearly, $L=V\setminus V(M_2)$ and $|L|=n-\ell$. Moreover, for each $i\in V$ we have both $i\in V(M_1)\cup L$ and $i\in V(M_2)\cup L$. Hence $(M_1,M_2,L)\in \bigcap_{i\in [2n]} A_{w,\ell,i}$. Here, for different $(M_1,M_2),(M'_1,M'_2)\in\wdpm$ it is obvious that the corresponding triplets $(M_1,M_2,L)$ and $(M'_1,M'_2,L')$ in $\bigcap_{i\in [2n]} A_{w,\ell,i}$ are different. Hence $|\wdpm|\leq \left| \bigcap_{i\in [2n]} A_{w,\ell,i} \right|$. This completes the proof.
\end{proof}

Using a well-known variation of the inclusion-exclusion principle, expressing the size of the intersection by an alternating sum of sizes of intersections of complements (cf.~\cite{DBLP:books/sp/CyganFKLMPPS15}), we arrive at the following expression for $|\wdpm|$.
\begin{align}
 |\wdpm|= \left| \bigcap_{i\in [2n]} A_{w,\ell,i} \right| = \sum_{I\subseteq[2n]} (-1)^{|I|} \left| \bigcap_{i\in I} \overline{A}_{w,\ell,i} \right|
 \label{math:ie:wdmp}
\end{align}

\section{Computing $\boldsymbol{|\wdpm|}$ for graphs of small treedepth}
\label{section:computingwdpm}

In this section, we show how to compute $|\wdpm|$ in $4^{\tau}\cdot n^{\Oh(1)}$ time and polynomial space when the edge weights are polynomially bounded and given an elimination forest $\T$ of depth $\tau$ for the graph. Our algorithm relies on equation $(\ref{math:ie:wdmp})$ for $|\wdpm|$. Since we do not have time proportional to $2^{2n}$ available, we will evaluate the contributions to the right-hand side while following the structure of $\T$. To this end, we will use polynomials in few indeterminates to conveniently handle different (partial) contributions like in previous work~\cite{DBLP:conf/stacs/HegerfeldK20,DBLP:journals/siamdm/NederlofPSW23}. Zoomed out, it looks a bit like the algorithm of Nederlof et al.~\cite{DBLP:journals/siamdm/NederlofPSW23} but we use inclusion-exclusion rather than cut-and-count for the main computation.

Throughout this section, let $G=(V,E)$ be a connected graph with vertex set $V=\{1,\ldots,n\}$, let $\weights\colon E\to \{1,\ldots,N\}$ with $N=n^{\Oh(1)}$, and let $w,\ell\in\naturals$ with $\ell$ even and $\ell\leq n$. Moreover, let $\T=(V,E_\T)$ be an elimination tree of $G$ of depth $\tau$ and with root $r$. We show how to compute $|\wdpm|$ in time $4^{\tau}\cdot n^{\Oh(1)}$ and polynomial space. It is straightforward to then use this approach to compute $|\wdpm|$ for disconnected graphs. (If so desired, the algorithm works for larger weights at the cost of a factor $N$ appearing in the running time.)

\subparagraph*{Additional notation.}
Some additional notation will be helpful to handle the fact that in equation $(\ref{math:ie:wdmp})$ there are two possible requirements for each vertex. Accordingly, for $W\subseteq V$ we let $W^+:=\{v,v+n \mid v\in W\}$. We also extend this to $\tail(v)$ etc., e.g., $\tail^+(v):=\{u,u+n \mid u\in\tail(v)\}$. Furthermore, we define $\tuple$, $\monomial$, and $\disjoint$, which are short for tuple, monomial, and disjoint:
\begin{align*}
 \tuple(W) &:= \{(E_1,E_2,L) \mid{} E_1\subseteq \delta[W], E_2\subseteq \delta[W], L\subseteq W, E_1\cap E_2 = \emptyset\}\\
 \monomial(E_1,E_2,L) &: =x^{|E_1|}y^{|E_2|}z^{|L|}\omega^{\weights(E_1\cup E_2)}\\
 \disjoint(I) &:= \{ (E_1,E_2,L) \mid{} \forall \{s,t\}\in E_1: \{s,t\}\cap I=\emptyset,\\
 & \hspace{2.7cm} \forall \{s,t\}\in E_2: \{s+n,t+n\}\cap I=\emptyset, L^+\cap I=\emptyset\}
\end{align*}
We will use, e.g., $\tuple(\tree(v))$ to succinctly refer to the possible parts of tuples in $U_{w,\ell}$ whose edges have at least one endpoint in $\tree(v)$; the weight is tracked through the arising polynomials (to be defined in a moment). Using $\monomial(E_1,E_2,L)$ we get a monomial whose degrees track the cardinalities of $E_1$, $E_2$, and $L$ as well as the total weight of edges in $E_1\cup E_2$ (which will always be disjoint). Using $(E_1,E_2,L)\in\disjoint(I)$ we express whether the tuple in question would be eligible for contributing for a given $I\subseteq \broom^+[v]$, which corresponds to a selection of requirements in equation $(\ref{math:ie:wdmp})$. (Note that $I\subseteq\broom^+[v]$ will always be implicit through combination of subsets of $\tail^+(v)$ and $\tree^+[v]$ resp.\ $\tail^+[v]$ and $\tree^+(v)$.)

\subparagraph*{Two polynomials.}
We are now ready to define two types of polynomials that will be used in our branching algorithm. For $v\in V$ and $J\subseteq \tail^+(v)$ let $P_{(v)}(J)\in \integers[x,y,z,\omega]$ with 
\begin{align*}
 P_{(v)}(J)=\sum_{\substack{(E_1,E_2,L)\in\\\tuple(\tree[v])}} \monomial(E_1,E_2,L) \sum_{\substack{I=J\cup K\\K\subseteq\tree^+[v]}}(-1)^{|K|} \iverson{(E_1,E_2,L)\in\disjoint(I)}
\end{align*}
Similarly, for $v\in V$ and $J'\subseteq \tail^+[v]$ let $Q_{[v]}(J')\in\integers[x,y,z,\omega]$ with
\begin{align*}
 Q_{[v]}(J')=\sum_{\substack{(E'_1,E'_2,L')\in\\\tuple(\tree(v))}} \monomial(E'_1,E'_2,L') \sum_{\substack{I'=J'\cup K'\\K'\subseteq\tree^+(v)}}(-1)^{|K'|} \iverson{(E'_1,E'_2,L')\in\disjoint(I')}
\end{align*}
Intuitively, these capture the contributions of partial tuples based on $\tree[v]$ resp.\ $\tree(v)$. The inner sums mimic the shape of equation $(\ref{math:ie:wdmp})$. The difference between the two types is that $P_{(v)}(J)$ uses $J\subseteq\tail^+(v)$, so does not specify a choice of $v$ and/or $v+n$ to be included in $J$, whereas this is included in $Q_{[v]}(J)$ with $J\subseteq\tail^+[v]$. It is crucial that the edges with exactly one endpoint in $\tree[v]$ resp.\ $\tree(v)$ have their other endpoint in $\tail(v)$ resp.\ $\tail[v]$ so that consistency with $I\subseteq\broom^+[v]$ can be checked.

We now establish recurrences and base cases that hold for these polynomials. We begin with the base case of $Q_{[v]}(J')$ for leaf nodes $v$ of $\T$. Then we prove recurrences that relate $Q_{[v]}(J)$ to $P_{(u_1)}(J),\ldots,P_{(u_o)}(J)$ for internal nodes $v$ with $\children(v)=\{u_1,\ldots,u_o\}$, and for relating $P_{(v)}(J)$ to $Q_{[v]}(J')$ (for arbitrary nodes). We then show how to get $|\wdpm|$ from $P_{(r)}(\emptyset)$. When spelling out the algorithm, it is crucial to work directly with the coefficients of the polynomials, but here their definitions via sums over tuples are more convenient.

\subparagraph*{Base case for $\boldsymbol{Q_{[v]}}$ for leaf nodes.}
Let $v$ be a leaf node of $\T$ and let $J'\subseteq\tail^+[v]$. Since $\tree(v)=\emptyset$ we have $\tuple(\tree(v))=\{(\emptyset,\emptyset,\emptyset)\}$ so the outer sum has only one summand. Similarly, $\tree^+(v)=\emptyset$, so $K'=\emptyset$ is the only choice for the inner sum. Clearly, for the unique choice $(\emptyset,\emptyset,\emptyset)\in\tuple(\tree(v))$ we get the summand $\monomial(\emptyset,\emptyset,\emptyset)=x^0y^0z^0\omega^0=1$ since $(-1)^{|K'|}=1$ and $(\emptyset,\emptyset,\emptyset)\in\disjoint(I')$. Thus, $Q_{[v]}(J')=1$.

\subparagraph*{Recurrence for $\boldsymbol{Q_{[v]}}$ for internal nodes.}
Let $v$ be an internal node of $\T$, let $J'\subseteq\tail^+[v]$, and let $\children(v)=\{u_1,\ldots,u_o\}$. We show that $Q_{[v]}(J')$ is simply the product of the polynomials $P_{(u_i)}(J')$ of the child nodes of $v$. (In hindsight this includes the case of leaf nodes as a special case, but it is good to treat that one explicitly.)

\begin{lemma}
 \label{lemma:qfromp}
 \begin{align*}
  Q_{[v]}(J')=\prod_{i\in[o]} P_{(u_i)}(J').
 \end{align*}
\end{lemma}

\begin{proof}
 We start from the right-hand side of the equation and as a first step plug in the definition of $P_{(u_i)}(J')$:
 \begin{align*}
      & \prod_{i\in[o]} P_{(u_i)}(J')\\
  ={} & \prod_{i\in[o]} \sum_{\substack{(E^i_1,E^i_2,L^i)\in\\\tuple(\tree[u_i])}} \monomial(E^i_1,E^i_2,L^i) \sum_{\substack{I^i=J'\cup K^i\\K^i\subseteq\tree^+[u_i]}}(-1)^{|K^i|} \iverson{(E^i_1,E^i_2,L^i)\in\disjoint(I^i)}
 \intertext{We rewrite the expression by expanding out the product over $i\in[o]$:}
  ={} & \sum_{\substack{(E^1_1,E^1_2,L^1)\in\tuple(\tree[u_1])\vspace{-1.8mm}\\ \vdots\\(E^o_1,E^o_2,L^o)\in\tuple(\tree[u_o])}} \monomial(E^1_1,E^1_2,L^1)\cdot\ldots\cdot\monomial(E^o_1,E^o_2,L^o)\\
  & \quad \cdot \sum_{\substack{I^1=J'\cup K^i\\K^1\subseteq\tree^+[u_1]\vspace{-1.8mm}\\\vdots\\I^o=J'\cup K^o\\K^o\subseteq\tree^+[u_o]}}(-1)^{|K^1|+\ldots+|K^o|} \iverson{(E^1_1,E^1_2,L^1)\in\disjoint(I^1)}\ldots\iverson{(E^o_1,E^o_2,L^o)\in\disjoint(I^o)}
 \end{align*}
 Note that the sets $\tree^+[u_i]$, with $i\in[o]$, are a partition of $\tree^+(v)$ since $\children(v)=\{u_1,\ldots,u_o\}$. We use this along with the structure of an elimination tree and the definitions of $\tuple(\tree[u_i])$ and $\disjoint(I^i)$ to see that we arrive at the definition of $Q_{[v]}(J')$.
 
 We begin with replacing the inner sum by the summation over $I'=J'\cup K'$ with $K'\subseteq\tree^+(v)$ as it appears in $Q_{[v]}(J')$. The following arguments show that this is correct:
 \begin{itemize}
  \item Intuitively, we replace $K^1,\ldots,K^o$ with $K^i\subseteq\tree^+[u_i]$ by $K'=K^1\cup\ldots\cup K^o\subseteq\tree^+(v)$. Since the sets $\tree^+[u_i]$ are a partition of $\tree^+(v)$ this is one-to-one. Moreover, we have $(-1)^{|K^1|+\ldots+|K^o|}=(-1)^{|K'|}$.
  \item We replace $\iverson{(E^1_1,E^1_2,L^1)\in\disjoint(I^1)}\ldots\iverson{(E^o_1,E^o_2,L^o)\in\disjoint(I^o)}$, which takes value $1$ exactly when all $o$ conditions are fulfilled (and zero otherwise) by 
  \[
   \iverson{(E^1_1\cup\ldots\cup E^o_1,E^1_2\cup\ldots\cup E^o_2,L^1\cup\ldots\cup L^o)\in\disjoint(I')}
  \]
  where, as intended, $I'=J'\cup K'=J'\cup K^1\cup\ldots\cup K^o$. We check that these two terms give the same value, which is equivalent to checking that the conditions in brackets hold if and only if (recall that, in general, $\iverson{p\wedge q}=\iverson{p}\iverson{q}$):
  \begin{itemize}
   \item Assume first that $(E^1_1\cup\ldots\cup E^o_1,E^1_2\cup\ldots\cup E^o_2,L^1\cup\ldots\cup L^o)\in\disjoint(I')$ holds. By definition, this requires that $\forall\{s,t\}\in E^1_1\cup\ldots\cup E^o_1: \{s,t\}\cap I'=\emptyset$. Since $E^i_1\subseteq E^1_1\cup\ldots\cup E^o_1$ and $I^i=J'\cup K^i\subseteq J'\cup K'=I'$, it follows that $\forall\{s,t\}\in E^i_1: \{s,t\}\cap I^i=\emptyset$. The argument for $\forall\{s,t\}\in E^i_2:\{s+n,t+n\}\cap I^i=\emptyset$ is analogous. Finally, it is also required that $(L^1\cup\ldots\cup L^o)^+\cap I'=\emptyset$. Since $L^i\subseteq L^1\cup\ldots\cup L^o$ and $I^i=J'\cup K^i\subseteq J'\cup K'=I'$, it follows that $(L^i)^+\cap I^i=\emptyset$. Thus, $(E^i_1,E^i_2,L^i)\in\disjoint(I^i)$ for all $i\in[o]$.
   \item Conversely, assume that $(E^i_1,E^i_2,L^i)\in\disjoint(I^i)$ for all $i\in[o]$. Here it is crucial that the sets $\tree[u_i]$ for $i\in[o]$ are a partition of $\tree(v)$, and ditto for using $\tree^+$. 
   
   (1) We know that $\forall\{s,t\}\in E^i_1: \{s,t\}\cap I^i=\emptyset$; recall that $I^i=J'\cup K^i$. Let $j\neq i$ and let $\{s,t\}\in E^i_1$. Then $\{s,t\}\in\delta[\tree[u_o]]$ since $(E^i_1,E^i_2,L^i)\in\tuple(\tree[u_i])$. W.l.o.g., let $s\in\tree[u_i]$. Now $t\in\tree[u_i]$ too or $t\notin\tree[u_i]$. In the latter case, since $\T$ is an elimination $\T$ and $t$ cannot be a descendant of $s$ (or else $t\in\tree[s]\subseteq\tree[u_i]$), node $t$ must be an ancestor of $s$; then $t\in\tail(s)\setminus\tree[u_i]=\tail(u_i)=\tail[v]$. Either way, we have $K^j\subseteq\tree^+[u_j]$, so $K^j$ is disjoint from $\tree^+[u_i]\supseteq\tree[u_i$. It is also disjoint from $\tail(u_j)=\tail[v]$, so it contains neither $s$ nor $t$. Thus, $\forall\{s,t\}\in E^i_1\forall j\neq i: \{s,t\}\cap K^j=\emptyset$. It follows that $\forall\{s,t\}\in E^i_1: \{s,t\}\cap I'=\emptyset$, since $I'=J'\cup K^1\cup\ldots\cup K^o$. And then this holds for all $\{s,t\}\in E^1_1\cup\ldots\cup E^o_i$.
   
   (2) The argument for getting $\forall\{s,t\}\in E^1_2\cup\ldots\cup E^o_2: \{s+n,t+n\}\cap I'=\emptyset$ is fully analogous. 
   
   (3) We know that $(L^i)^+\cap I^i=\emptyset$; recall that $I^i=J'\cup K^i$. Since $(E^i_1,E^i_2,L^i)\in\tuple(\tree[u_i])$, we have $L^i\subseteq\tree[u_i]$ and, hence, $(L^i)^+\subseteq\tree^+[u_i]$. Now let $j\neq i$. Then $K^j\subseteq\tree^+[u_j]$ hence $(L^i)^+\cap K_j\subseteq \tree^+[u_i]\cap\tree^+[u_j]=\emptyset$, as the sets $\tree^+[u_i]$ are a partition of $\tree^+(v)$. It follows that $(L^i)^+\cap I'= (L^i)^+\cap(J'\cup K_1\cup\ldots\cup K^o)=\emptyset$. Hence $(L^1\cup\ldots\cup L^o)^+ \cap I'=((L^1)^+\cup\ldots\cup (L^o)^)+\cap I'=\emptyset$.
   
   Thus $(E^1_1\cup\ldots\cup E^o_1,E^1_2\cup\ldots\cup E^o_2,L^1\cup\ldots\cup L^o)\in\disjoint(I')$.
  \end{itemize}
 \end{itemize}
 Applying this replacement we arrive at
 \begin{align*}
  ={} & \sum_{\substack{(E^1_1,E^1_2,L^1)\in\tuple(\tree[u_1])\vspace{-1.8mm}\\ \vdots\\(E^o_1,E^o_2,L^o)\in\tuple(\tree[u_o])}} \monomial(E^1_1,E^1_2,L^1)\cdot\ldots\cdot\monomial(E^o_1,E^o_2,L^o)\\
  & \quad \cdot \sum_{\substack{I'=J'\cup K'\\K'\subseteq\tree^+(v)}}(-1)^{|K'|} \iverson{(E^1_1\cup\ldots\cup E^o_1,E^1_2\cup\ldots\cup E^o_2,L^1\cup\ldots\cup L^o)\in\disjoint(I')}
 \end{align*}
 Now we work on the outer summation. Since the inner sum only depends on, e.g., $E^1_1\cup\ldots\cup E^o_1$ but not on the separate sets, we can let the outer summation go over $(E_1,E_2,L)\in\tuple(v)$. The following arguments show that this is correct:
 \begin{itemize}
  \item The sets $\tree[u_1],\ldots,\tree[u_o]$ are a partition of $\tree(v)$ since $\children(v)=\{u_1,\ldots,u_o\}$. Moreover, there are no edges with endpoints in any two different sets $\tree[u_i]$ and $\tree[u_j]$, with $i\neq j$, by definition of an elimination tree. Thus, the sets $\delta[\tree[u_1]],\ldots,\delta[\tree[u_o]]$ are a partition of $\delta[\tree(v)]$.
  \item First, this gives a canonical bijection between choices of $o$ sets $L^1\subseteq\tree[u_1],\ldots,L^o\subseteq\tree[u_o]$ and single subsets $L\subseteq\tree(v)$: Take union from former to latter, take projections $\tree[u_1],\ldots,\tree[u_o]$ from letter to former.
  \item Second, we get a similar bijection between choice of $o$ pairs of disjoint subsets $E^1_1,E^1_2\subseteq\delta[\tree[u_1]],\ldots E^o_1,E^o_2\subseteq\delta[\tree[u_o]]$ and single pairs of disjoint subsets $E_1,E_2\subseteq\delta[\tree(v)]$; union one way, projections the other.
  \item Together, we get a bijection between choices of $o$ tuples $(E^1_1,E^1_2,L^1)\in\tuple(\tree[u_1]),\ldots,(E^o_1,E^o_2,L^o)\in\tuple(\tree[u_o])$ and single tuples $(E_1,E_2,L)\in\tuple(\tree(v))$. Thus, we get the same number of summands whether we go by the $o$ tuples or by the single tuple. It remains to check that the values are the same along our bijection.
  \item Compare $\monomial(E^1_1,E^1_2,L^1)\cdot\ldots\cdot\monomial(E^o_1,E^o_2,L^o)$ and $\monomial(E_1,E_2,L)$ for corresponding choices. It is easy to see that the product of $o$ monomials on the left is exactly the single monomial on the right. E.g., we have that $E^1_1,\ldots,E^o_1$ are a partition of $E_1$, so either way we arrive at $x^{|E_1|}$ in the resulting monomial.
  \item In the inner sum, it is clear that nothing changes when we replace, e.g., $E^1_1\cup\ldots\cup E^o_1$ by $E_1$, since these are identical (along our bijection).
 \end{itemize}
 With this second replacement we arrive at the definition of $Q_{[v]}(J')$, as desired.
 \begin{align*} 
  ={} & \sum_{\substack{(E_1,E_2,L)\in\\\tuple(\tree(v))}} \monomial(E_1,E_2,L) \sum_{\substack{I'=J'\cup K'\\K'\subseteq\tree^+(v)}}(-1)^{|K'|} \iverson{(E_1,E_2,L)\in\disjoint(I)}\\
  ={} & Q_{[v]}(J')
 \end{align*}
 This completes the proof.
\end{proof}

\subparagraph*{Recurrence for $\boldsymbol{P_{(v)}}$.}
Now let $v$ be an arbitrary node of $\T$ and let $J\subseteq\tail^+(v)$. We show that $P_{(v)}(J)$ can be expressed using sums and products of polynomials, mainly certain $Q_{[v]}(J')$. Intuitively, going from $Q_{[v]}(J')$ to $P_{(v)}(J)$ drops the fixed behavior of $v$ (as given by $J'\subseteq\tail^+[v]$) and brings into scope the edges incident with $v$ that were not previously considered, i.e., the set $E(v,\tail(v))$. In the recurrence, the contributions of these edges, along with the contribution for $v$ and $v+n$, have to be accounted for. Intuitively, factors such as $(1+z\iverson{J_v=\emptyset})$ correspond to the option of adding, e.g., vertex $v$ to a tuple: If not added, then this leaves the corresponding monomials unchanged (factor $1$). If added (which requires some condition in $\iverson{}$ to hold, in this case $J_v=\emptyset$) then this also creates modified monomials corresponding to this (factor $z$ in this case). Similar factors correspond to the edges $e\in E(v,\tail(v))$.

\begin{lemma}
 \label{lemma:pfromq}
 \begin{align*}
  P_{(v)}(J)={} & \sum_{\substack{J'=J\cup J_v\\J_v\subseteq\{v,v+n\}}} (-1)^{|J_v|} \cdot Q_{[v]}(J') \cdot (1 + z \iverson{J_v=\emptyset})\\
  & \cdot \prod_{\substack{e=\{u,v\}\\ e\in E(v,\tail(v))}} 1+x\omega^{\weights(e)}\iverson{\{u,v\}\cap J'=\emptyset}+y\omega^{\weights(e)}\iverson{\{u+n,v+n\}\cap J'=\emptyset}
 \end{align*}
\end{lemma}

\begin{proof}
 We first show how to rewrite part of the right-hand side using $\monomial$ and $\disjoint$.

 \begin{claim}
  For $J_v\subseteq\{v,v+n\}$ and $J'=J\cup J_v$ it holds that
  \begin{align*}
   & (1 + z \iverson{J_v=\emptyset}) \smash[b]{\prod_{\substack{e=\{u,v\}\\ e\in E(v,\tail(v))}}} 1+x\omega^{\weights(e)}\iverson{\{u,v\}\cap J'=\emptyset}\\
   & \hspace{7cm} +y\omega^{\weights(e)}\iverson{\{u+n,v+n\}\cap J'=\emptyset}\\[0.5em]
   ={} & \sum_{\substack{\hat{E}_1\subseteq E(v,\tail(v))\\\hat{E}_2\subseteq E(v,\tail(v))\\\hat{L}\subseteq\{v\}\\\hat{E}_1\cap\hat{E}_2=\emptyset}} \monomial(\hat{E}_1,\hat{E}_2,\hat{L}) \cdot \iverson{(\hat{E}_1,\hat{E}_2,\hat{L}) \in \disjoint(J')}.
  \end{align*}
 \end{claim}

 \begin{claimproof}
  Intuitively, the RHS of the claimed equation is obtained by expanding out the product on the LHS but without gathering multiple copies of the same monomial $x^ay^bz^c\omega^d$ that may arise from $\monomial(\hat{E}_1,\hat{E}_2,\hat{L})$ over all choices of $\hat{E}_1$, $\hat{E}_2$, and $\hat{L}$ that fulfill $(\hat{E}_1,\hat{E}_2,\hat{L})\in\disjoint(J')$. Instead, there are separate summands for all combinations of choices.
  
  $LHS\leq RHS$: Consider a single occurrence of a monomial $x^ay^bz^c\omega^d$ in the expansion of the product on the LHS after dropping summands whose condition, e.g., $\{u,v\}\cap J'=\emptyset$, is not fulfilled (so they get factor $\iverson{\{u,v\}\cap J'=\emptyset}=0$). Clearly $c\in\{0,1\}$. Let $\hat{L}=\{v\}$ if $c=1$ and $\hat{L}=\emptyset$ otherwise (if $c=0$); either way, $\hat{L}\subseteq\{v\}$. Observe that we can only have $c=1$ if $J_v=\emptyset$ holds, or else $\iverson{J_v=\emptyset}=0$. Let $\hat{E}_1\subseteq E(v,\tail(v))$ contain those edges $e\in E(v,\tail(v))$ where $x\omega^{\weights(e)}$ was contributed to the monomial, which requires that $\{u,v\}\cap J'=\emptyset$ for these edges. Similarly, let $\hat{E}_2\subseteq E(v,\tail(v)$ contain those edges $e\in E(v,\tail(v))$ where $y\omega^{\weights(e)}$ was contributed to the monomial, which requires that $\{u+n,v+n\}\cap J'=\emptyset$ for these edges. Clearly, $\hat{E}_1\cap\hat{E}_2=\emptyset$ since for each $e\in E(v,\tail(v))$ only one choice out of $1$, $x\omega^{\weights(e)}$, and $y\omega^{\weights(e)}$ contributes to the monomial in the expansion of the product. Furthermore, since these are all possible contributions, we then have $x^ay^bz^c\omega^d=x^{|\hat{E}_1|}y^{|\hat{E}_2|}z^{|\hat{L}|}\omega^{\weights(\hat{E}_1\cup\hat{E}_2)}=\monomial(\hat{E}_1,\hat{E}_2,\hat{L})$. It remains to verify that $(\hat{E}_1,\hat{E}_2,\hat{L}) \in \disjoint(J')$, so that $\iverson{(\hat{E}_1,\hat{E}_2,\hat{L}) \in \disjoint(J')}=1$ and we get the same contribution of $\monomial(\hat{E}_1,\hat{E}_2,\hat{L})$ on the RHS (for this specific choice of $\hat{E}_1$, $\hat{E}_2$, and $\hat{L}$). The following conditions are required (and showed to hold) in order for $(\hat{E}_1,\hat{E}_2,\hat{L}) \in \disjoint(J')$; recall that $J'=J\cup J_v$ and $J\subseteq\tail^+(v)$:
  \begin{itemize}
   \item $\forall\{s,t\}\in\hat{E}_1: \{s,t\}\cap J'=\emptyset$: Since $\hat{E}_1\subseteq E(v,\tail(v))$, let $\{s,t\}=\{u,v\}$.  For all these edges we already observed that $\{u,v\}\cap J'=\emptyset$ for all edges in $\hat{E}_1$.
   \item $\forall\{s,t\}\in\hat{E}_2: \{s+n,t+n\}\cap J'=\emptyset$: Similarly, let $\{s,t\}=\{u,v\}$. We similarly already observed that $\{u+n,v+n\}\cap J'=\emptyset$ for all edges in $\hat{E}_2$.
   \item $\hat{L}^+\cap J'=\emptyset$: Since $\hat{L}\subseteq\{v\}$ we have $\hat{L}^+\subseteq\{v,v+n\}$. Hence, $\hat{L}^+\cap J'=\hat{L}^+\cap J_v$. Recall that we can only have $c=1$ and $\hat{L}=\{v\}$ when $J_v=\emptyset$; in this case $\hat{L}^+\cap J_v=\emptyset$. Otherwise, we have $\hat{L}=\emptyset$ and, hence, $\hat{L}^+=\emptyset$, so again $\hat{L}^+\cap J_v=\emptyset$.
  \end{itemize}
  We see that for each single occurrence of any monomial in the expansion of the product on the LHS there is a unique (and private) choice of $\hat{E}_1$, $\hat{E}_2$, and $\hat{L}$ that contributes the same monomial as a summand to the RHS. Since there are no negative terms on the RHS it follows that $LHS\leq RHS$, as claimed.
  
  $RHS\leq LHS$: Consider a nonzero summand on the RHS and let $\hat{E}_1,\hat{E}_2\subseteq E(v,\tail(v))$ with $\hat{E}_1\cap\hat{E}_2$ and $\hat{L}\subseteq\{v\}$ be the corresponding choices. We have $(\hat{E}_1,\hat{E}_2,\hat{L}) \in \disjoint(J')$ (as required for a nonzero summand). The summand hence is 
  \[
   \monomial(\hat{E}_1,\hat{E}_2,\hat{L})=x^{|\hat{E}_1|}y^{|\hat{E}_2|}z^{|\hat{L}|}\omega^{\weights(\hat{E}_1\cup\hat{E}_2)}.
  \]
  Using $(\hat{E}_1,\hat{E}_2,\hat{L}) \in \disjoint(J')$ we show that there is the same monomial occurring in the expansion of the product of the LHS for a combination of choices that is unique to $(\hat{E}_1,\hat{E}_2,\hat{L})$:
  \begin{itemize}
   \item From $(1+z\iverson{J_v=\emptyset})$ choose $1$ if $\hat{L}=\emptyset$ and choose $z$ if $\hat{L}=\{v\}$. In the latter case, we indeed have $\iverson{J_v=\emptyset}=1$: From $(\hat{E}_1,\hat{E}_2,\hat{L}) \in \disjoint(J')$ we get $\emptyset=\hat{L}^+\cap J'\supseteq \{v,v+n\}\cap J_v$, which is only possible when $J_v=\emptyset$ (as $J_v\subseteq\{v,v+n\}$). In both cases, this matches the contribution of $\hat{L}$ to $\monomial(\hat{E}_1,\hat{E}_2,\hat{L})=x^{|\hat{E}_1|}y^{|\hat{E}_2|}z^{|\hat{L}|}\omega^{\weights(\hat{E}_1\cup\hat{E}_2)}$, namely $z^{|\hat{L}|}$.
   \item Now consider the factor $1+x\omega^{\weights(e)}\iverson{\{u,v\}\cap J'=\emptyset}+y\omega^{\weights(e)}\iverson{\{u+n,v+n\}\cap J'=\emptyset}$ for some edge $e=\{u,v\}\in E(v,\tail(v)$. We choose depending on which, if any, of $\hat{E}_1$ and $\hat{E}_2$ contains $e$; recall that $\hat{E}_1\cap\hat{E}_2=\emptyset$:
   \begin{itemize}
    \item If $e\notin\hat{E}_1$ and $e\notin\hat{E}_2$ then choose $1$. This matches the (trivial) contribution of $e$ to $\monomial(\hat{E}_1,\hat{E}_2,\hat{L})=x^{|\hat{E}_1|}y^{|\hat{E}_2|}z^{|\hat{L}|}\omega^{\weights(\hat{E}_1\cup\hat{E}_2)}$.
    \item If $e\in\hat{E}_1$ and $e\notin\hat{E}_2$ then choose $x\omega^{\weights(e)}$. Let us check that $\iverson{\{u,v\}\cap J'=\emptyset}=1$: From $(\hat{E}_1,\hat{E}_2,\hat{L}) \in \disjoint(J')$ we get $\forall\{s,t\}\in\hat{E}_1: \{s,t\}\cap J'=\emptyset$, so indeed $\{u,v\}\cap J'=\emptyset$. Again, $x\omega^{\weights(e)}$ matches the contribution of $e$ to $\monomial(\hat{E}_1,\hat{E}_2,\hat{L})$. 
    \item If $e\notin\hat{E}_1$ and $e\in\hat{E}_2$ then choose $y\omega^{\weights(e)}$. Let us check that $\iverson{\{u+n,v+n\}\cap J'=\emptyset}=1$: From $(\hat{E}_1,\hat{E}_2,\hat{L}) \in \disjoint(J')$ we get $\forall\{s,t\}\in\hat{E}_2: \{s+n,t+n\}\cap J'=\emptyset$, so indeed $\{u+n,v+n\}\cap J'=\emptyset$. Also here, $y\omega^{\weights(e)}$ matches the contribution of $e$ to $\monomial(\hat{E}_1,\hat{E}_2,\hat{L})$. 
    \item It is not possible that $e\in\hat{E}_1$ and $e\in\hat{E}_2$.
   \end{itemize}
  \end{itemize}
  Clearly, if we consider two tuples $(\hat{E}_1,\hat{E}_2,\hat{L}) \neq (\hat{E}'_1,\hat{E}'_2,\hat{L}')$ corresponding to summands on the RHS then above we make a different choice from $(1+z\iverson{J_v=\emptyset})$, if $\hat{L}\neq\hat{L}'$, or from $1+x\omega^{\weights(e)}\iverson{\{u,v\}\cap J'=\emptyset}+y\omega^{\weights(e)}\iverson{\{u+n,v+n\}\cap J'=\emptyset}$, for all $e$ with $e\in\hat{E}_1\triangle\hat{E}'_1$ and/or $e\in\hat{E}_2\triangle\hat{E}'_2$. (This does not exclude getting the same monomial, but ensures that each summand corresponds to a unique term in the expansion.)
  
  Overall, for each nonzero summand on the RHS we find the same monomial as a unique (and private) term in the expansion of the product on the LHS. Since there are no negative terms, we get that $RHS\leq LHS$.
  
  Thus, $LHS=RHS$, as claimed.
 \end{claimproof}

 Using the claim, the right-hand side of the equation in the lemma statement can be written more concisely:
 \begin{align*}
  & \sum_{\substack{J'=J\cup J_v\\J_v\subseteq\{v,v+n\}}} (-1)^{|J_v|} \cdot Q_{[v]}(J') \sum_{\substack{\hat{E}_1\subseteq E(v,\tail(v))\\\hat{E}_2\subseteq E(v,\tail(v))\\\hat{L}\subseteq\{v\}\\\hat{E}_1\cap\hat{E}_2=\emptyset}} \monomial(\hat{E}_1,\hat{E}_2,\hat{L}) \cdot \iverson{(\hat{E}_1,\hat{E}_2,\hat{L}) \in \disjoint(J')}
 \end{align*} 
 We plug in the definition of $Q_{[v]}(J')$:
 \begin{align*}
  ={} & \sum_{\substack{J'=J\cup J_v\\J_v\subseteq\{v,v+n\}}} (-1)^{|J_v|} \sum_{\substack{(E'_1,E'_2,L')\in\\\tuple(\tree(v))}} \monomial(E'_1,E'_2,L') \sum_{\substack{I'=J'\cup K'\\K'\subseteq\tree^+(v)}}(-1)^{|K'|} \cdot \iverson{(E'_1,E'_2,L')\in\disjoint(I')}\\
  & \cdot \sum_{\substack{\hat{E}_1\subseteq E(v,\tail(v))\\\hat{E}_2\subseteq E(v,\tail(v))\\\hat{L}\subseteq\{v\}\\\hat{E}_1\cap\hat{E}_2=\emptyset}} \monomial(\hat{E}_1,\hat{E}_2,\hat{L}) \cdot \iverson{(\hat{E}_1,\hat{E}_2,\hat{L})\in\disjoint(J')}
 \end{align*} 
 We reorder summation:
 \begin{align*}
  ={} & \sum_{\substack{(E'_1,E'_2,L')\in\\\tuple(\tree(v))}} \sum_{\substack{\hat{E}_1\subseteq E(v,\tail(v))\\\hat{E}_2\subseteq E(v,\tail(v))\\\hat{L}\subseteq\{v\}\\\hat{E}_1\cap\hat{E}_2=\emptyset}} \monomial(E'_1,E'_2,L') \cdot \monomial(\hat{E}_1,\hat{E}_2,\hat{L}) \\
  & \sum_{\substack{J'=J\cup J_v\\J_v\subseteq\{v,v+n\}}} \sum_{\substack{I'=J'\cup K'\\K'\subseteq\tree^+(v)}} (-1)^{|J_v|} \cdot (-1)^{|K'|} \cdot \iverson{(E'_1,E'_2,L')\in\disjoint(I')} \cdot \iverson{(\hat{E}_1,\hat{E}_2,\hat{L})\in\disjoint(J')}
 \end{align*} 
  We merge the inner two summations and consolidate the products in the summands: First, replace $J_v\subseteq\{v,v+n\}$ and $K'\subseteq\tree^+(v)$ by $K=J_v \cup K'\subseteq\tree^+[v]$, using that $\tree^+[v]=\tree^+(v)\cup\{v,v+n\}$, and replace $(-1)^{|J_v|} \cdot (-1)^{|K'|}$ by $(-1)^{|K|}$. We use $I:=J\cup K$ to reflect the change, but note that $I'=J'\cup K'=J\cup J_v \cup K'=J\cup K=I$. Finally, replace $\iverson{(E'_1,E'_2,L')\in\disjoint(I')} \cdot \iverson{(\hat{E}_1,\hat{E}_2,\hat{L})\in\disjoint(J')}$ by $\iverson{(E'_1\cup\hat{E}_1,E'_2\cup\hat{E}_2,L'\cup\hat{L})\in\disjoint(I)}$, which we justify now:
  \begin{itemize}
   \item Since $I'=I$ and $J'\subseteq I$, we have that $(E'_1\cup\hat{E}_1,E'_2\cup\hat{E}_2,L'\cup\hat{L})\in\disjoint(I)$ implies $(E'_1,E'_2,L')\in\disjoint(I')$ and $(\hat{E}_1,\hat{E}_2,\hat{L})\in\disjoint(J')$.
   \item Conversely, if $(E'_1,E'_2,L')\in\disjoint(I')$ and $(\hat{E}_1,\hat{E}_2,\hat{L})\in\disjoint(J')$ then (with a bit more work) we also get that $(E'_1\cup\hat{E}_1,E'_2\cup\hat{E}_2,L'\cup\hat{L})\in\disjoint(I)$:
   \begin{itemize}
    \item $\forall\{s,t\}\in E'_1\cup\hat{E}_1: \{s,t\}\cap I=\emptyset$: We have this for all edges in $E'_1$ from $(E'_1,E'_2,L')\in\disjoint(I')$ since $I'=I$. Let $\{s,t\}\in \hat{E}_1\subseteq E(v,\tail(v))$; then $s,t\in\tail[v]$ and, hence, $\{s,t\}\cap\tree^+(v)=\emptyset$. Hence, $\{s,t\}\cap I= \{s,t\}\cap (J'\cup K')\subseteq \{s,t\}\cap (J'\cup \tree^+(v))=\{s,t\}\cap J'=\emptyset$, where the final equality holds as $(\hat{E}_1,\hat{E}_2,\hat{L})\in\disjoint(J')$.
    \item $\forall\{s,t\}\in E'_2\cup\hat{E}_2: \{s+n,t+n\}\cap I=\emptyset$: Fully analogous to previous item.
    \item $(L'\cup\hat{L})^+\cap I=\emptyset$: Clearly $(L'\cup\hat{L})^+=L'^+\cup\hat{L}^+$. We have $L'^+\cap I=\emptyset$ since $I=I'$ and $(E'_1,E'_2,L')\in\disjoint(I')$. Furthermore, since $\hat{L}\subseteq\{v\}$ and $K'\subseteq\tree^+(v)$, we have $\hat{L}^+\cap I=\hat{L}^+\cap (J'\cup K')=\hat{L}^+\cap J'=\emptyset$ as $(\hat{E}_1,\hat{E}_2,\hat{L})\in\disjoint(J')$.
   \end{itemize}
  \end{itemize}
  With this replacement we arrive at the following:
 \begin{align*}
  ={} & \sum_{\substack{(E'_1,E'_2,L')\in\\\tuple(\tree(v))}} \sum_{\substack{\hat{E}_1\subseteq E(v,\tail(v))\\\hat{E}_2\subseteq E(v,\tail(v))\\\hat{L}\subseteq\{v\}\\\hat{E}_1\cap\hat{E}_2=\emptyset}} \monomial(E'_1,E'_2,L') \cdot \monomial(\hat{E}_1,\hat{E}_2,\hat{L}) \\
  & \hspace{4cm} \cdot \sum_{\substack{I=J\cup K\\K\subseteq\tree^+[v]}} (-1)^{|K'|} \cdot \iverson{(E'_1\cup\hat{E}_1,E'_2\cup\hat{E}_2,L'\cup\hat{L})\in\disjoint(I)}
 \end{align*}
  
  Finally, we merge the outer two summations and consolidate the product preceding the inner summation. Intuitively, we want to merge $E_1=E'_1\cup\hat{E}_1$, $E_2=E'_2\cup\hat{E}_2$, and $L=L'\cup\hat{L}$, and let the (new) outer summation go over $(E_1,E_2,L)\in\tuple(\tree[v])$. We show that this yields the same summands, hence the same result:
  \begin{itemize}
   \item Let $(E'_1,E'_2,L')\in\tuple(\tree(v))$, which entails $E'_1,E'_2\subseteq\delta[\tree(v)]$, $L'\subseteq\tree(v)$, and $E'_1\cap E'_2=\emptyset$. Further, let $\hat{E}_1\subseteq E(v,\tail(v))$, $\hat{E}_2\subseteq E(v,\tail(v))$, and $\hat{L}\subseteq\{v\}$ with $\hat{E}_1\cap\hat{E}_2=\emptyset$. Using these, let $E_1=E'_1\cup\hat{E}_1$, $E_2=E'_2\cup\hat{E}_2$, and $L=L'\cup\hat{L}$. We show that $(E_1,E_2,L)$ is in $\tuple(\tree[v])$:
   \begin{itemize}
    \item $E_1\subseteq\delta[\tree[v]]$: Since $E'_1\subseteq\delta[\tree(v)]$ and $\hat{E}_1\subseteq E(v,\tail(v))$, it follows that $E_1=E'_1\cup\hat{E}_1\subseteq\delta[\tree[v]]$.
    \item $E_2\subseteq\delta[\tree[v]]$: Analogous.
    \item $L\subseteq\tree[v]$: Since $L'\subseteq\tree(v)$ and $\hat{L}\subseteq\{v\}$, it follows that $L=L'\cup\hat{L}\subseteq\tree[v]$.
    \item $E_1\cap E_2=\emptyset$: Since $\tree(v)\cap\tail[v]=\tree(v)\cap(\{v\}\cup\tail(v))=\emptyset$, we have $\delta[\tree(v)]\cap E(v,\tail(v))=\emptyset$. Thus, $E'_i\cap \hat{E}_j=\emptyset$ for $i,j\in\{1,2\}$. We also have that $E'_1\cap E'_2=\emptyset$ and $\hat{E}_1\cap\hat{E}_2=\emptyset$. Hence, $E_1\cap E_2=(E'_1\cup\hat{E}_1)\cap (E'_2\cup\hat{E}_2)=\emptyset$.
   \end{itemize}
   \item Conversely, let $(E_1,E_2,L)\in\tuple(\tree[v])$. We show that there are partitions $E_1=E'_1\cup\hat{E}_1$, $E_2=E'_2\cup \hat{E}_2$, and $L=L'\cup\hat{L}$ such that $(E'_1,E'_2,L')\in\tuple(\tree(v))$, $\hat{E}_1\subseteq E(v,\tail(v))$, $\hat{E}_2\subseteq E(v,\tail(v))$, $\hat{L}\subseteq\{v\}$, and $\hat{E}_1\cap\hat{E}_2\cap=\emptyset$. Concretely, let $E'_1=E_1\cap \delta[\tree(v)]$ and $\hat{E}_1=E_1\setminus E'_1$. Similarly, let $E'_2=E_2\cap \delta[\tree(v)]$ and $\hat{E}_2=E_2\setminus E'_2$. Finally, let $L'=L\cap \tree(v)$ and $\hat{L}=L\setminus L'$.
   
   We check that $(E'_1,E'_2,L')\in\tuple(\tree(v))$:
   \begin{itemize}
    \item By definition we have $E'_1,E'_2\subseteq\delta[\tree(v)]$ and $L'\subseteq\tree(v)$.
    \item Since $E_1\cap E_2=\emptyset$ and $E'_1\subseteq E_1 $ and $E'_2\subseteq E_2$, we have $E'_1\cap E'_2=\emptyset$.
   \end{itemize}
   
   We check that $\hat{E}_1\subseteq E(v,\tail(v))$, $\hat{E}_2\subseteq E(v,\tail(v))$, $\hat{L}\subseteq\{v\}$, and $\hat{E}_1\cap\hat{E}_2=\emptyset$:
   \begin{itemize}
    \item $\hat{E}_1\subseteq E(v,\tail(v))$: We have $\hat{E}_1= E_1\setminus E'_1=E_1\setminus \delta[\tree(v)]$ and $E_1\subseteq\delta[\tree[v]]$. Thus, $\hat{E}_1$ only contains edges with at least one endpoint in $\tree[v]$ but no endpoint in $\tree(v)$. Hence all edges therein have $v$ as an endpoint and (still) no endpoint in $\tree(v)$. By the definition of an elimination tree, the other endpoint of such edges must be an ancestor of $v$, i.e., must be in $\tail(v)$. Thus, $\hat{E}_1\subseteq E(v,\tail(v))$.
    \item $\hat{E}_2\subseteq E(v,\tail(v))$: Analogous.
    \item $\hat{L}\subseteq\{v\}$: Since $L\subseteq\tree[v]=\tree(v)\cup\{v\}$ and $L'=L\cap\tree(v)$, we have $\hat{L}=L\setminus L'=L\setminus\tree(v)\subseteq \tree[v]\setminus\tree(v)=\{v\}$.
    \item Since $E_1\cap E_2=\emptyset$ and $\hat{E}_1\subseteq E_1 $ and $\hat{E}_2\subseteq E_2$, we have $\hat{E}_1\cap \hat{E}_2=\emptyset$.
   \end{itemize}

   \item It is easy to see that the above gives a bijection between $(E_1,E_2,L)\in\tuple(\tree[v])$ and the combinations of choices in the ranges of the outer two sums. Forward we have a projection of each component into two disjoint sets, e.g., $E_1$ into $E'_1=E_1\cap\delta(\tree(v)\subseteq\delta(\tree(v))$ and $\hat{E}_1=E_1\setminus E'_1\subseteq E(v,\tail(v))$. Backward we take unions of sets that are subsets of disjoint sets, so different pairs of sets cannot have the same union.
   
   \item It remains to check that along this bijection we get the same summands. Clearly the inner sums yield the same because they only depend on, e.g., $E_1=E'_1\cup\hat{E}_1$ not on the separate sets. It remains to consider $\monomial(E'_1,E'_2,L') \cdot \monomial(\hat{E}_1,\hat{E}_2,\hat{L})$, using the disjointness observed above, e.g., $E'_i\cap \hat{E}_j=\emptyset$ for $i,j\in\{1,2\}$:
    \begin{align*}
     \monomial(E'_1,E'_2,L') \cdot \monomial(\hat{E}_1,\hat{E}_2,\hat{L}) &= x^{|E'_1|}y^{|E'_2|}z^{|L'|}\omega^{\weights(E'_1\cup E'_2)} \cdot x^{|\hat{E}_1|} y^{|\hat{E}_2|} z^{|\hat{L}|} \omega^{\weights(\hat{E}_1\cup\hat{E}_2)}\\
     &= x^{|E'_1\cup\hat{E}_1|} y^{|E'_2\cup\hat{E}_2|} z^{|L'\cup\hat{L}|} \omega^{\weights((E'_1\cup\hat{E}_1)\cup(E'_2\cup\hat{E}_2))}\\
     &=\monomial(E'_1\cup\hat{E}_1,E'_2\cup\hat{E}_2,L'\cup\hat{L})\\
     &=\monomial(E_1,E_2,L)
    \end{align*}
   \end{itemize}  
  This completes the argument, and we arrive at the following after making the replacements:
 \begin{align*}
  ={} & \sum_{\substack{(E_1,E_2,L)\in\\\tuple(\tree[v])}} \monomial(E_1,E_2,L) \sum_{\substack{I=J\cup K\\K\subseteq\tree^+[v]}}(-1)^{|K|} \iverson{(E_1,E_2,L)\in\disjoint(I)}\\
  ={} & P_{(v)}(J)
 \end{align*}
 
 This completes the proof.
\end{proof}

\subparagraph*{Getting $\boldsymbol{|\wdpm|}$ from $\boldsymbol{P_{(r)}(\emptyset)}$.}
We now show how to get $|\wdpm|$ from $P_{(r)}(J)$, specifically from $P_{(r)}(\emptyset)$ since that is the only possible choice for $J\subseteq\tail^+(r)=\emptyset$. Essentially, we just need a specific coefficient of $P_{(r)}(\emptyset)$ depending on $w$ and $\ell$.

\begin{lemma}
 \label{lemma:wdpmfrompremptyset}
 Let $P_{(r)}(\emptyset)=\sum_{a,b,c,d} \alpha_{a,b,c,d}\cdot x^ay^bz^c\omega^d$. Then $|\wdpm|=\alpha_{\ell/2,\ell/2,n-\ell,w}$.
\end{lemma}

\begin{proof}
 We recall the definition of $P_{(v)}(J)$ for $v\in V$ and $J\subseteq\tail^+(v)$ as
 \begin{align*}
  P_{(v)}(J)=\sum_{\substack{(E_1,E_2,L)\in\\\tuple(\tree[v])}} \monomial(E_1,E_2,L) \sum_{\substack{I=J\cup K\\K\subseteq\tree^+[v]}}(-1)^{|K|} \iverson{(E_1,E_2,L)\in\disjoint(I)}.
 \end{align*}
 Generally for $P_{(v)}(J)$, each $(E_1,E_2,L)\in\tuple(\tree[v])$ contributes $\alpha_{(E_1,E_2,L)}\cdot \monomial(E_1,E_2,L)=\alpha_{(E_1,E_2,L)}\cdot x^{|E_1|} y^{|E_2|} z^{|L|} \omega^{\weights(E_1\cup E_2)}$ where
 \begin{align*}
  \alpha_{(E_1,E_2,L)} = \sum_{\substack{I=J\cup K\\K\subseteq\tree^+[v]}}(-1)^{|K|} \iverson{(E_1,E_2,L)\in\disjoint(I)}.
 \end{align*}
 Thus, any specific coefficient $\alpha_{a,b,c,d}$ is equal to the sum of $\alpha_{(E_1,E_2,L)}$ over all $(E_1,E_2,L)\in\tuple(\tree[v])$ with $\monomial(E_1,E_2,L)=x^ay^bz^c\omega^d$.
 
 We now consider $P_{(r)}(\emptyset)$, i.e., $v=r$ and $\emptyset=J\subseteq\tail^+(r)=\emptyset$, and apply this for $\alpha_{\ell/2,\ell/2,n-\ell,w}$:
 \begin{align*}
  \alpha_{\ell/2,\ell/2,n-\ell,w} &= \sum_{\substack{(E_1,E_2,L)\in\tuple(\tree[r])\\\monomial(E_1,E_2,L)=x^{\ell/2}y^{\ell/2}z^{n-\ell}\omega^w}} \sum_{\substack{I=J\cup K\\K\subseteq\tree^+[v]}}(-1)^{|K|} \iverson{(E_1,E_2,L)\in\disjoint(I)}\\
  &= \sum_{\substack{(E_1,E_2,L)\in\tuple(V)\\\monomial(E_1,E_2,L)=x^{\ell/2}y^{\ell/2}z^{n-\ell}\omega^w}} \sum_{I\subseteq [2n]}(-1)^{|I|} \iverson{(E_1,E_2,L)\in\disjoint(I)}
 \end{align*}
 Note that $\tree[r]=V$, $\tree^+[r]=V^+=[2n]$, and $I=K$ in the inner sum as $J=\emptyset$.
 
 We observe that the outer sum is in fact over $(E_1,E_2,L)\in U_{w,\ell}$: 
 \begin{itemize}
  \item Recall that
  \begin{align*}
   U_{w,\ell}:= \left\{ (E_1,E_2,L) \mid E_1,E_2\in \binom{E}{\ell/2}, E_1\cap E_2=\emptyset, \weights(E_1\cup E_2)=w, L\in\binom{V}{n-\ell} \right\}.
  \end{align*}
  \item Recall also the definition of $\tuple(W)$, which can be simplified using $W=V$,
  \begin{align*}
   \tuple(V) &=\{(E_1,E_2,L) \mid{} E_1\subseteq \delta[V], E_2\subseteq \delta[V], L\subseteq V, E_1\cap E_2 = \emptyset\}\\
   &= \{(E_1,E_2,L) \mid{} E_1\subseteq E, E_2\subseteq E, L\subseteq V, E_1\cap E_2 = \emptyset\}.
  \end{align*}
  \item Finally, by definition of $\monomial(E_1,E_2,L)$, the second restriction of the sum is equivalent to $|E_1|=|E_2|=\frac\ell2$, $|L|=n-\ell$, and $\weights(E_1\cup E_2)=w$. Hence
  \begin{align*}
   \alpha_{\ell/2,\ell/2,n-\ell,w}&= \sum_{(E_1,E_2,L)\in U_{w,\ell}} \sum_{I\subseteq [2n]}(-1)^{|I|} \iverson{(E_1,E_2,L)\in\disjoint(I)}
  \end{align*}
 \end{itemize}

 We now observe that $(E_1,E_2,L)\in\disjoint(I)$ holds for $(E_1,E_2,L)\in U_{w,\ell}$ if and only if $(E_1,E_2,L)\in \bigcap_{i\in I} \overline{A}_{w,\ell,i}$: 
 \begin{itemize}
  \item Recall that $A_{w,\ell,i}\subseteq U_{w,\ell}$ for $i\in[2n]$ with
  \begin{align*}
   A_{w,\ell,i}:= \begin{cases}
                 \left\{ (E_1,E_2,L)\in U_{w,\ell} \mid i\in V(E_1)\cup L \right\} & \mbox{if $i\in[n]$}\\
                 \left\{ (E_1,E_2,L)\in U_{w,\ell} \mid i-n\in V(E_2)\cup L \right\} & \mbox{if $i\in[n+1,2n]$.}
                \end{cases}
  \end{align*}
  \item Hence, $(E_1,E_2,L)\in \overline{A}_{w,\ell,i}$ is equivalent to $i\notin V(E_1)\cup L$ if $i\in[n]$ and to $i-n\notin V(E_2)\cup L$ if $i\in[2n]$. Then $(E_1,E_2,L)\in \bigcap_{i\in I} \overline{A}_{w,\ell,i}$ if and only if this holds for all $i\in I$. This in turn is equivalent to $(E_1,E_2,L)\in\disjoint(I)$, which was defined as
 \begin{align*}
  \disjoint(I):=\{ (E_1,E_2,L) \mid{} & \forall \{s,t\}\in E_1: \{s,t\}\cap I=\emptyset,\\
  & \forall \{s,t\}\in E_2: \{s+n,t+n\}\cap I=\emptyset,\\
  & L^+\cap I=\emptyset\}.
 \end{align*}
 \end{itemize}
 
 This allows us to further transform our equation, starting with changing the order of summation:
 \begin{align*}
  \alpha_{\ell/2,\ell/2,n-\ell,w}&= \sum_{I\subseteq [2n]}(-1)^{|I|} \sum_{(E_1,E_2,L)\in U_{w,\ell}} \iverson{(E_1,E_2,L)\in\disjoint(I)}\\
  &= \sum_{I\subseteq [2n]}(-1)^{|I|} \sum_{(E_1,E_2,L)\in U_{w,\ell}} \iverson{(E_1,E_2,L)\in \bigcap_{i\in I} \overline{A}_{w,\ell,i}}\\
  &= \sum_{I\subseteq [2n]}(-1)^{|I|} \left| \bigcap_{i\in I} \overline{A}_{w,\ell,i} \right|\\
  &= |\wdpm|
 \end{align*}
 This completes the proof. 
\end{proof}

\subparagraph*{The algorithm.}
While it was convenient for proving the recurrences to define $P_{(v)}(J)$ and $Q_{[v]}(J')$ as sums over suitable triples (of essentially partial solutions), our algorithm unsurprisingly uses their representation via coefficients as polynomials in $\integers[x,y,z,\omega]$. (Also we just saw that the coefficients carry the information that we are after, e.g., $|\wdpm|$.) Similar to previous work (e.g.,~\cite{DBLP:conf/stacs/HegerfeldK20}), our algorithm proceeds by recursively computing the two types of polynomials using the established recurrences. We nevertheless spell out the pseudocode to more easily establish time and space complexity. To this end, we also roll the computation of $Q_{[v]}(J')$ into the computation of $P_{(v)}(J)$ rather than making this a separate recursive call.

\begin{algorithm}[H]
 \KwIn{$v\in V$ and $J\subseteq\tail^+(v)$}
 \KwOut{$P_{(v)}(J)$} 
 $P=0$\;
 \For(\tcp*[f]{outer sum in Lemma~\ref{lemma:pfromq}}){$J_v\subseteq\{v,v+n\}$}
 {
  $J'=J\cup J_v$\;
  $summand=(-1)^{|J_v|}$\;
  \lIf{$J_v=\emptyset$}{$summand=summand\cdot (1+z)$}
  \For(\tcp*[f]{summand for inner product}){$e=\{u,v\}\in E(v,\tail(v))$}{
   $factor=1$\;
   \lIf{$\{u,v\}\cap J'=\emptyset$}{$factor=factor+x\omega^{\weights(e)}$}
   \lIf{$\{u+n,v+n\}\cap J'=\emptyset$}{$factor=factor+y\omega^{\weights(e)}$}
   $summand=summand\cdot factor$\;
  }
  
  \tcp{multiplication with $Q_{[v]}(J')$, step by step using Lemma~\ref{lemma:qfromp}, implicit multiplication by $Q_{[v]}(J')=1$ if $v$ is a leaf}
  \For{$u\in\children(v)$}
  {
   $summand=summand\cdot \ComputeP(u,J')$\;
  }
  $P=P+ summand$\;
 }
 \Return{$P$}\;
 \caption{\ComputeP($v$,$J$)}
 \label{algo:computep}
\end{algorithm}

\subparagraph*{Time and space complexity.}
All polynomials are represented via their coefficients as polynomials in $\integers[x,y,z,\omega]$. It can be checked that the degrees of $x$ and $y$ are bounded by $m=|E|$, the degree of $z$ is bounded by $n$, and the degree of $\omega$ is bounded by the total weight $\weights(E)$ of all edges, which is at most $m\cdot N$. Using the definitions of $P_{(v)}(J)$ and $Q_{[v]}(J')$ it can also be checked that the absolute value of each coefficient is at most $2^m\cdot 2^m\cdot 2^n\cdot 4^n$ (by using trivial upper bounds for the number of terms in the summation), which gives polynomial-size binary encoding per coefficient. Overall, for $N$ polynomial in the input size, we use polynomial-size per polynomial. The recursion depth is bounded by $\tau$ and each call internally uses at most three polynomials (namely $P$, $summand$, and $factor$). The space complexity is therefore bounded by $\Oh(\tau\cdot n^{\Oh(1)})$ when $N$ is polynomially bounded.

Clearly, each single call to \ComputeP can be implemented to run in polynomial time (optionally with fast polynomial multiplication, especially when using a large weight bound $N$) when not counting the recursion. Observe that for each node $v$ there are $2^{\tail^+(v)}\leq 4^{\tau}$ possible calls involving $v$ and it can be verified that there are no repeated calls. Thus, the algorithm runs in time $4^{\tau}n^{\Oh(1)}$.

\begin{theorem}
 \label{theorem:algorithm:wdpm}
 Given a graph $G=(V,E)$ with $n$ vertices, edge weights $\weights\colon E\to\{1,\ldots,N\}$ with $N=n^{\Oh(1)}$, values $w,\ell\in\naturals$, and an elimination forest $\T$ for $G$ of depth $\tau$, the value $|\wdpm|$ can be computed in time $4^{\tau}\cdot n^{\Oh(1)}$ and polynomial space.
\end{theorem}

\begin{proof}
 If $\ell>n$, or $\ell$ not even, or $w>n\cdot N$ we can return $0$ as the correct answer. If $G$ is disconnected then a simple DP lets us compute for all $i$, $w'$, and $\ell'$ the value $|\mathcal{M}_{w',\ell'}|$ for the graph induced by the first $i$ connected components of $G$ by using the algorithm for the connected case on each connected component (and suitable values $\hat{w},\hat{\ell}$). This causes only a polynomial in $n$ overhead in the running time. The algorithm for the connected case and even $\ell\leq n$ was explained above.
\end{proof}

\section{An algorithm for partial cycle cover}
\label{section:applications}

Using our algorithm for computing $|\wdpm|$, given graph $G=(V,E)$, $w,\ell\in\naturals$, and an elimination forest $\T$ of depth $\tau$ of $G$, only little work remains for solving \partialcyclecover. Similar to cut and count we also rely on symmetries of (non-)solutions having too many cycles plus the isolation lemma to (probabilistically) guarantee a unique solution (if one exists). Differently (and closer to Nederlof et al.~\cite{DBLP:journals/siamdm/NederlofPSW23} who also using matchings) we use that the edges of an even cycle have two ways of being partitioned into an ordered pair of disjoint matchings covering the same vertex sets (i.e., ordered pairs of consistent matchings). Thus, non-solutions with too many even cycles can be filtered out by counting modulo a suitable power of $2$. Since partial cycle covers may also contain odd cycles, however, we need to first use a simple trick to reduce to the bipartite case, where the above idea is sufficient.

\begin{lemma}
 \label{lemma:reducetobipartite}
 Let $G$ be a graph, let $k,\ell\in\naturals$, and let $\T$ be an elimination forest of depth $\tau$ of $G$. Let $G'$ be obtained from $G$ by subdividing all its edges and let $\ell'=2\ell$. Then $G$ has the same number of partial cycle covers with at most $k$ cycles covering exactly $\ell$ vertices as $G'$ has partial cycle covers with at most $k$ cycles covering exactly $\ell'$ vertices.
Moreover, an elimination tree $\T'$ of depth at most $\tau+1$ of $G'$ can be efficiently constructed from $\T$.
\end{lemma}

\begin{proof}
 Formally, let $G=(V,E)$ and let $G'=(V',E')$ with $V'=V\cup\{p_e\mid e\in E\}$ and $E'=\{\{p_e,v\} \mid e\in E, v\in e\}$. Let $\C_{k,\ell}$ denote the family of edge sets of partial cycle covers with at most $k$ cycles and exactly $\ell$ vertices in $G$. Similarly, let $\C'_{k,\ell'}$ denote the same for $G'$, with at most $k$ cycles and exactly $\ell'=2\ell$ vertices. There is a one-to-one correspondence between elements of $\C_{k,\ell}$ and those of $\C'_{k,\ell'}$: Simply do respectively undo the edge subdivisions on the edges present in the respective partial cycle cover. (E.g., a cycle $(u,v,w,u)$ in $G$ corresponds to $(u,p_{\{u,v\}},v,p_{\{v,w\}},w,p_{\{w,u\}},u)$ in $G'$ with exactly twice the length, and vice versa.) It follows that $|\C_{k,\ell}|=|\C'_{k,\ell'}|$, as claimed.
 
 To see the moreover part, take $\T$ and for each $e=\{u,v\}\in E$ attach the subdividing vertex $p_e$ as a new child node to whoever of $u$ and $v$ is lower in $\T$; crucially, $u$ and $v$ must be in ancestor-descendant-relation in $\T$ due to the edge $\{u,v\}$ in $G$. In this way, for each edge $\{s,p_{\{s,t\}}\}\in E'$ node $s$ must be an ancestor of $p_{\{s,t\}}$. Call this tree $\T'$. Clearly, the depth of $\T'$ is at most $\tau+1$.
\end{proof}

Keeping the lemma in mind, the following theorem gives an algorithm for \partialcyclecover on bipartite graphs. The case of general graphs will be an immediate corollary.

\begin{theorem}
 \label{theorem:partialcyclecover:bipartite}
 There is a randomized algorithm that, given a bipartite graph $G$, numbers $k,\ell\in\naturals$, and an elimination forest $\T$ of depth $\tau$ of $G$ solves \partialcyclecover in $4^{\tau}n^{\Oh(1)}$ time and polynomial space. The algorithm makes no false positives and may make a false negative with probability at most $\frac12$.
\end{theorem}

\begin{proof}
 We again use $\C_{k,\ell}$ for the family of edge sets of partial cycles covers in $G$ with at most $k$ cycles and on exactly $\ell$ vertices. The algorithm needs to discover (up to the error chance) whether $C_{k,\ell}\neq\emptyset$. Since $G$ is bipartite, all occurring cycles are of even length.
 
 Let $n=|V|$ and $m=|E|$. Choose $\weights\colon E\to\{1,\ldots,N\}$ with $N=2|E|=\Oh(n^2)$ uniformly and independently at random. For each $w\in\{1,\ldots,\ell N\}$, compute $|\wdpm|$ in $4^{\tau}n^{\Oh(1)}$ time and polynomial space. Output \yes if $|\wdpm|\not\equiv 0 \mod 2^{k+1}$ for any $w$; otherwise output \no. This completes the algorithm. Time and space complexity should be clear, let us argue absence of false positives and bound the probability of false negatives.
 
 We begin with some observations about the family $\wdpm$; let us recall its definition as
 \begin{align*}
  \wdpm:= \{ (M_1,M_2) \mid{} & \mbox{$M_1$ and $M_2$ are consistent matchings in $G$ of cardinality $\frac{\ell}2$ each}, \\
  & V(M_1)=V(M_2), \weights(M_1\cup M_2)=w \}.
 \end{align*}
 For each $(M_1,M_2)\in\wdpm$, the two matchings $M_1$ and $M_2$ are consistent. Hence their union is a partial cycle cover (with each vertex in $V(M_1)=V(M_2)$ having degree two and all other vertices having degree zero). Obviously $|V(M_1)|=|V(M_2)|=\ell=2|M_1|=2|M_2|$, i.e., these partial cycle covers visit exactly $\ell$ vertices each. Similar to previous work~\cite{DBLP:journals/siamdm/NederlofPSW23} we use the symmetry in partial cycles covers with more than $k$ cycles for cancellation. Presently, a partial cycle cover with exactly $p$ even cycles on exactly $\ell$ vertices corresponds to $2^p$ ordered pairs of consistent matchings on (the same) $\ell$ vertices: Let $C_1,\ldots,C_p$ denote the edge sets of the $p$ even cycles. For each $C_i$ there are two choices for putting half of its edges into the first matching and the other half into the second. Overall, this yields $2^p$ ordered pairs of consistent matchings per edge set of any partial cycle cover with $p$ (even) cycles. Let us add a few observations about these pairs:
 \begin{itemize}
  \item They all consist of two consistent matchings of cardinality $\frac{\ell}2$ each.
  \item The weight of each pair is equal to the weight of the corresponding partial cycle cover.
  \item Thus, if any pair $(M_1,M_2)$ is in $\wdpm$ then all $2^p$ pairs corresponding to the partial cycle cover $M_1\cup M_2$ are contained in \wdpm.
 \end{itemize}
 It follows that the contribution of partial cycle covers on $\ell$ vertices but with more than $k$ cycles is congruent $0$ modulo $2^{k+1}$, i.e., it does not affect the outcome of the algorithm.
 
 Let us assume first that $G$ has no partial cycle cover of at most $k$ cycles on exactly $\ell$ vertices. Then $|\C_{k,\ell}|=0$. Since every pair $(M_1,M_2)\in\wdpm$ corresponds to a partial cycle cover on exactly $\ell$ vertices in $G$ but not contained in $\C_{k,\ell}=\emptyset$, all those partial cycle covers have more than $k$ cycles. But then $|\wdpm|\equiv 0 \mod 2^{k+1}$ so the algorithm will return $\no$ irrespective of the (random) choice of $\weights$. In other words, there are no false positives.

 Let us now assume that $G$ has at least one partial cycle cover of at most $k$ cycles on exactly $\ell$ vertices. Then $|\C_{k,\ell}|\geq 1$. By the isolation lemma it follows that with probability at least $\frac12$ there is a unique element of $\C_{k,\ell}$ of minimum weight. Assuming that this succeeds, we show that the algorithm will return \yes, so a false negative as probability at most $\frac12$. Let $C^*\in\C_{k,\ell}$ be of unique minimum weight $w=\weights(C^*)$. Thus, in $\wdpm$ we have two types of pairs $(M_1,M_2)$ of consistent matchings:
 \begin{itemize}
  \item Those pairs that correspond to $C^*$, i.e., with $M_1\cup M_2=C^*$. There are exactly $2^p$ of them, where $p\leq k$ is the number of cycles in $C^*$. Their number is not congruent to $0$ modulo $2^{k+1}$.
  \item All other pairs $(M_1,M_2)$. As observed above, their union is a partial cycle cover on exactly $\ell$ vertices (and weight $w$). They do not correspond to $C^*$, which is unique among these partial cycle covers with weight $w$ and at most $k$ cycles, so their corresponding partial cycle covers have more than $k$ cycles. We have already discussed that their number is congruent $0$ modulo $2^{k+1}$.
 \end{itemize}
 Thus, assuming successful isolation (which has probability at least $\frac12$) the algorithm discovers $|\wdpm|\not\equiv 0 \mod 2^{k+1}$ and returns \yes. Thus, returning \no, a false negative, has probability at most $\frac12$. This completes the proof.
\end{proof}

Using Lemma~\ref{lemma:reducetobipartite}, the extension to general graphs is now an easy corollary.

\begin{corollary}
 The algorithmic result of Theorem~\ref{theorem:partialcyclecover:bipartite} also holds for general graphs $G$, all other assumptions being the same.
\end{corollary}

\begin{proof}
 Let $G=(V,E)$ be a (general) graph, $k,\ell\in\naturals$, and $\T$ an elimination forest of depth $\tau$ of $G$. Construct the graph $G'$ as in Lemma~\ref{lemma:reducetobipartite} so that the number of partial cycle covers with at most $k$ cycles on exactly $\ell$ vertices in $G$ is equal to the number of partial cycle covers with at most $k$ cycles on exactly $\ell'=2\ell$ vertices in $G'$. Let $\T'$ be the obtained elimination forest of depth at most $\tau+1$ for $G'$; in short $|\C_{k,\ell}|=|\C'_{k,\ell'}|$. Since $G'$ is bipartite, we can use the algorithm of Theorem~\ref{theorem:partialcyclecover:bipartite} to determine whether $\C'_{k,\ell'}\neq \emptyset$ (up to error bounds), and return the same answer for $G$ and $k,\ell\in\naturals$. Since $G'$ is only polynomially larger than $G$, and since $\T'$ has depth at most $\tau+1$, we get the same asymptotic bounds of $4^{\tau}n^{\Oh(1)}$ time and polynomial space (and same one-side error bound).
\end{proof}

Altogether, this constitutes a proof of Theorem~\ref{theorem:main:partialcyclecover} in the introduction. The transfer to the problems listed in Corollary~\ref{corollary:main:applications} follows as explained by Nederlof et al.~\cite{DBLP:journals/siamdm/NederlofPSW23}.

\section{Conclusion}
\label{section:conclusion}

We have given a $4^{\tau}n^{\Oh(1)}$ time and polynomial space algorithm for solving the \partialcyclecover problem on graphs given with an elimination forest of depth at most $\tau$, which directly improves the time bounds for several related problems obtained in previous work~\cite{DBLP:journals/siamdm/NederlofPSW23}. It is worth noting that the dependence on the treedepth is now the same as that relative to treewidth, but the treewidth-based DP algorithm relies on an instance of cut-and-count that does not seem to transfer easily to a branching algorithm relative to treedepth (cf.~\cite{DBLP:conf/stacs/HegerfeldK20,DBLP:journals/siamdm/NederlofPSW23}). Intuitively, the DP computes path packings with each path being red or blue (a known perspective of cut-and-count), using the symmetries in possible colorings of disconnected solutions, but the role of vertices changes as we update partial solutions. This does not seem to work in the branching, polynomial-space setting used (so far) relative to treedepth.

With the $4^{\tau}n^{\Oh(1)}$ time and polynomial space, e.g., for \hamiltoniancycle, it is natural to ask about further improvements. Notably, the conditional lower bound of $(2+\sqrt{2}-\varepsilon)^\pw n^{\Oh(1)}$ (i.e., relative to pathwidth)~\cite{DBLP:journals/jacm/CyganKN18} has no implications for the treedepth parameterization. Nevertheless, it is a natural and interesting question whether $(2+\sqrt{2})^\tau n^{\Oh(1)}$ time and polynomial space are possible when given an elimination forest of depth $\tau$. If even possible, this should be quite challenging as the DP algorithm relative to pathwidth is a much more involved version of that relative to treewidth. (To note, this DP could be adapted to work on tree decompositions as well, but the lack of a fast operation at join nodes makes it worse than $4^\tw n^{\Oh(1)}$ time.)
Generally, it is very interesting what (conditionally optimal) algorithmic results can be transferred from treewidth and pathwidth to work in the same time but polynomial space relative to treedepth. Similarly, much (if not more) remains to do regarding (conditional) lower bounds relative to treedepth.

\bibliography{hc_td.bib}

\end{document}